\documentclass{article}

\usepackage[left=0.75in,right=0.75in,bottom=1in,top=1in]{geometry}
\usepackage[table]{xcolor} 
\usepackage{multicol}
\usepackage{titling}
\usepackage{amsthm}
\usepackage{abstract}
\usepackage{authblk}
\usepackage{makecell}
\usepackage{natbib}
\usepackage[colorlinks=false]{hyperref}
\usepackage{float}
\usepackage{natbib}
\usepackage{graphicx}
\usepackage{amsmath,amssymb}
\usepackage{amsfonts}
\usepackage{tabu}
\usepackage{caption}
\usepackage{booktabs}
\usepackage{multirow}
\usepackage{comment}

\providecommand{\keywords}[1]
{
  \small	
  \textbf{\textit{Keywords---}} #1
}

\newcommand{\reals}{\mathbb{R}}

\author[1]{Alex Ziyu Jiang} 
\author[1]{Abel Rodriguez}
\affil[1]{Department of Statistics, University of Washington}
\title{Improvements on Scalable Stochastic Bayesian Inference Methods for Multivariate Hawkes Process}
\date{}

\begin{document}

\maketitle

\begin{abstract}
Multivariate Hawkes Processes (MHPs) are a class of point processes that can account for complex temporal dynamics among event sequences. In this work, we study the accuracy and computational efficiency of three classes of algorithms which, while widely used in the context of Bayesian inference, have rarely been applied in the context of MHPs: stochastic gradient expectation-maximization, stochastic gradient variational inference and stochastic gradient Langevin Monte Carlo. An important contribution of this paper is a novel approximation to the likelihood function that allows us to retain the computational advantages associated with conjugate settings while reducing approximation errors associated with the boundary effects. The comparisons are based on various simulated scenarios as well as an application to the study of risk dynamics in the Standard \& Poor's 500 intraday index prices among its 11 sectors.
\end{abstract}


\section{Introduction}

The multivariate Hawkes process (MHP) model \citep{hawkes1971,liniger2009multivariate} is a class of temporal point process models that can capture complex time-event dynamics among multiple objects. Specifically, MHPs demonstrate the \textit{self-} and \textit{mutually-exciting} properties in multidimensional event sequences, where an event occurrence in a certain dimension leads to a higher likelihood of future events appearing in the same or other dimensions. This feature of the models makes MHPs attractive in a wide range of applications, including earth sciences \citep{ogata1988statistical}, finance \citep{bacry2015hawkes} and social media analysis \citep{rizoiu2017hawkes}. %However, developing inference algorithms for MHP models is not a trivial task, due to the fact that the event sequences are temporally dependent. 

Computational methods for maximum likelihood inference in Hawkes process models include direct maximization of the likelihood function (e.g., see \citealp{ozaki1979maximum}) and the expectation-maximization (EM) algorithm (e.g., see \citealp{veen2008estimation} and \citealp{Lewis11}), {\color{black} 
as well as penalized least-squares estimation method  of \cite{bacry2020sparse} and the variable selection-integrated maximum likelihood method of \cite{bonnet2023inference} for sparse interaction settings. 
%For multivariate Hawkes process models with model sparsity, estimation methods include the penalized least-squares estimation method  of \cite{bacry2020sparse} and the variable selection-integrated maximum likelihood method of \cite{bonnet2023inference}.
In the context of Bayesian inference, some of the algorithms that have been proposed include Markov Chain Monte Carlo algorithms (MCMC) \citep{rasmussen2013bayesian,mohler2013modeling,holbrook2021scalable, holbrook2022viral, deutsch2022bayesian}, variational approximations \citep{xu2017dirichlet,malem2022variational, sulem2022scalable}, sequential Monte Carlo \citep{linderman2017bayesian}, and the maximum \textit{a posteriori} probability estimation using the Expectation-Maximization algorithm (EM) \citep{zhang2018efficient}. In addition, theoretical guarantees for nonparametric Bayesian estimation methods of MHPs have been studied in \cite{donnet2020nonparametric} and \cite{sulem2021bayesian}.} One key challenge associated with all these computational approaches is that they do not scale well to large datasets. Specifically, the double summation operation needed to carry out a single likelihood evaluation is typically of time complexity $\mathcal{O}(KN^2)$, where $K$ is the number of dimensions and $N$ is the number of total events. Even in situations where careful implementation can reduce the time complexity to $\mathcal{O}(KN)$ (e.g., for exponential excitation functions), the cost of this operation can be prohibitive for moderately large datasets.  Furthermore, for methods that utilize the branching structure of MHPs, the space complexity is $\mathcal{O}(N^2)$ in all cases. An additional complication is that the calculation of the so-called ``compensator'' term in the likelihood function might limit our ability to exploit potential conjugacy in the model structure.  Standard approximations to the compensator, which are well-justified when maximizing the full likelihood, can have a more serious impact when applied to small datasets, e.g., those arising in the context of algorithms that use subsets of the original data.    %These challenges severely limit the application of MHP models %quadratic complexity in computation and space has led to challenges in applying MHPs 
%to large real-world datasets.

Algorithms inspired by stochastic optimization \citep{Robbins51} ideas, which approximate the gradient of the objective function through noisy versions evaluated on subsamples, offer an alternative for Bayesian inference on large datasets. Examples of such algorithms include stochastic gradient EM algorithms for finding the posterior mode of a model (e.g., see \citealp{chen2018stochastic}), stochastic gradient variational algorithms (e.g., see \citealp{Hoffman13}) and stochastic gradient Hamiltonian Monte Carlo methods (e.g., see \citealp{nemeth2021stochastic} and references therein).  The use of stochastic gradient methods in the context of MHP models is, nonetheless, limited. Exceptions include \cite{linderman2015scalable}, who consider the use of stochastic gradient variational inference in the context of a discretized MHP, and \cite{nickel2020learning}, who discuss stochastic gradient methods to directly maximize the observed data likelihood.
%
%{\color{black} In the context of MHP models, \cite{linderman2015scalable} applied the stochastic gradient variational algorithms to a discrete version of MHP, and \cite{linderman2017bayesian} later developed a stochastic gradient EM algorithm for the continous version. Under the frequentist framework, \cite{nickel2020learning} also developed a stochastic gradient descent algorithm for situations where a large number of event sequences are observed. Nonetheless, the accuracy of these methods and the tradeoffs associated with their use have not been explored in the literature (see \citealp{zhou2020} for a comparison of MCMC, EM and variational approximation methods without using stochastic optimization).}
% 

In this paper, we discuss the efficient implementation of stochastic gradient EM, stochastic gradient variational approximations, and stochastic gradient Langevin diffusion methods in the context of parametric MHP models, and evaluate various aspects of their performance using both simulated and real datasets.  A key contribution is an investigation of a novel approximation technique for the likelihood of the subsamples based on first-order Taylor expansion of the compensator term of the MHP models. We show that this novel approximation can lead to improvements in both point and interval estimation accuracy. {\color{black} For illustration purposes, we focus on intensity functions with exponential excitation functions. However, the insights gained from our experiments can be useful when working with other excitation functions that are proportional to density functions for which a conjugate prior on the unknown parameters is tractable.}

Not only is the literature on stochastic gradient methods for Bayesian inference in MHP models limited, but the trade-offs between computational speed and accuracy are not well understood in this context.  For \textit{full-batch} methods (i.e., when using gradients based on the whole dataset rather than subsamples) \cite{zhou2020} compares the estimation properties for EM, variational and random-walk MCMC algorithms.  Our work extends this comparative evaluation to algorithms based on stochastic gradient methods. 
{\color{black} Additionally, we apply our methods to study the market risk dynamics in the Standard \& Poor (S\&P)'s 500 intraday index prices for its 11 sectors, using all three computation approaches. Our analysis suggests that the 11 sectors can be grouped into four categories. Price movements from sectors within the same category are more likely to be associated with each other, compared to sectors from different categories.  Our analysis also shows that the effective range of the interactions between sectors is in the order of at most a few minutes.}

%The remainder of the paper is structured as follows: we introduce the MHP model in Section 2 and discuss the related methods in Section 3. We discuss the results of simulation experiments in Section 4 and real-world applications in Section 5.

%%%%%%%%%%%%%%%%%%%%%%%%%%%%%%
\section{Multivariate Hawkes process models}\label{sec:mhp}

%Let $\mathbf{X} = \{t_i^k, k = 1, \dots, K, i = 1,\dots,n_k\}$ be a sequence of events observed on $K$ dimensions, 
Let $\mathbf{X} = \{(t_i,d_i) : i = 1,\dots,n\}$ be a realization from from a marked point process where $t_i \in \reals^{+}$ represents the time at which the $i$-th event occurs and $d_i \in\{1, \ldots, K\}$ is a mark that represents the dimension in which the event occurs.  For example, $t_i$ might represent the time at which user $d_i$ makes a social media post, or the time at which the price of stock $d_i$ drops below a certain threshold.
%Let $\mathbf{X} = \{t_i, i = 1,\dots,n\}$ be a sequence of events observed on $K$ dimensions, where $t_i$ represents the time at which the $i$-th event and $d_i \in\{1, \ldots, K\}$ represents its corresponding dimension.
Also, let $n$ be the total number of events in the sequence and let $n_k$ be the number of events in dimension $k$. Similarly, let $\mathcal{H}_t = \{(t_i, d_i): t_i < t, t_i \in \mathbf{X}\}$ be the set of events that happened up until time $t$, and $N^{(k)}(t)$ be the number of events in dimension $k$ that occurred on $[0,t]$.  A sequence $\mathbf{X}$ follows a multivariate Hawkes process if the conditional density function on dimension $\ell$ %,  defined as the expected rate of arrivals at time point $t$ given $\mathcal{H}_t$,
has the following form: 
\begin{align}
\label{eq:1}
\lambda_{\ell}(t) \equiv \lim_{h \rightarrow 0} \frac{\mathbb{E}[N^{(\ell)}(t+h)-N^{(\ell)}(t) \mid \mathcal{H}_t]}{h} = \mu_\ell+ \sum_{k=1}^K \sum_{t_{i}<t, d_i = k} \phi_{k,\ell}\left(t-t_{i}\right),
\end{align}
%
%
%\begin{align}
%\label{eq:1}
%\lambda_{k}(t) &\equiv \lim_{h \rightarrow 0} \frac{\mathbb{E}[N^{(k)}(t+h)-N^{(k)}(t) \mid \mathcal{H}_t]}{h} \\&=\mu_k+ \sum_{j=1}^K \sum_{t_{i}^j<t, t_{i}^j \in \mathbf{X}} \phi_{k,\ell}\left(t-t_{i}^j\right),
%\end{align}
%
where $\mu_\ell > 0$ is the background intensity for dimension $\ell$, and $\phi_{k,\ell}(\cdot): \mathbf{R}^+ \rightarrow \mathbf{R}^+$ is the {excitation function} that controls how previous events in dimension $\ell$ affect the occurrence of new events in dimension $k$. 

{\color{black} Common modeling choices for the excitation function include the exponential decay function where $\phi_{k,\ell}(\Delta) = \alpha_{k,\ell} \beta_{k,\ell} e^{-\beta_{k,\ell} \Delta}$ for $\Delta \ge 0$, and the power law decay function where where $\phi_{k,\ell}(\Delta) = \frac{\alpha_{k,\ell} \beta_{k,\ell}}{(1+\beta_{k,\ell} \Delta)^{1+\gamma_{k,\ell}}}$ for $\Delta \ge 0$.} For illustration purposes, in this paper we focus on the exponential decay function. In that case, the parameter $\alpha_{k,\ell}$ controls the importance of events from dimension $k$ on the appearance of events in dimension $\ell$, and $\beta_{k,\ell}$ controls the magnitude of exponential decay of the instant change associated with a new event.

Let $\boldsymbol{\theta} = (\boldsymbol{\alpha}, \boldsymbol{\beta}, \boldsymbol{\mu})$ denote the vector of all model parameters. Using standard theory for point processes (e.g., see \citealp{daley2008}), the observed log-likelihood associated with a Hawkes process can be written as %(see the supplementary material for the detailed expression):

\begin{equation}
\label{eq:obslik}
\begin{aligned}
\mathcal{L}({\mathbf{X} \mid \boldsymbol{\alpha}, \boldsymbol{\beta}, \boldsymbol{\mu}})&= \sum_{\ell=1}^K \sum_{d_i = \ell}  \log \lambda_{\ell}\left(t_{i}\right)- \sum_{\ell=1}^K\int_{0}^T \lambda_{\ell}(s) ds \\
&= \sum_{\ell=1}^K \sum_{d_i =\ell}  \log \left(\mu_\ell  +  \sum_{k=1}^K \sum_{\substack{j < i \\ d_{j} = k, d_{i} = \ell}}\alpha _{k,\ell}\beta _{k,\ell}e^{-\beta _{k,\ell}(t_i-t_{j})} \right)- \sum_{\ell=1}^K \mu_\ell T \\
& \;\;\;\;\;\;\;\;\;\;\;\;\;\;\;\;\;\;\;\;\;\;\;\;\;\;\;\;\;\;\;\;\;\;\;\;\;\; - \sum_{k=1}^K \sum_{\ell=1}^K \alpha _{k,\ell} \left[n_k- \sum_{d_i = k} \exp \left(-\beta _{k,\ell}\left(T-t_{i}\right)\right)\right].
\end{aligned}    
\end{equation}

The MHP can also be obtained as a multidimensional Poisson cluster process  in which each point is considered an ``inmigrant'' or an ``offspring''  \citep{hawkes1974cluster,marsan2008,zhou2013,rasmussen2013bayesian}. We use the lower-triangular $n \times n$ binary matrix $\mathbf{B}$ to represent the latent branching structure of the events, where each row contains one and only one non-zero entry. For the strictly lower-triangular entries on the matrix, $B_{ij} = 1$ indicates that the $i$-th event can be viewed as an offspring of the $j$-th event. On the other hand, for the diagonal entries of the matrix, $B_{ii} = 1$ indicates that the $i$-th event is an immigrant. Each immigrant independently generates a cluster of offsprings that can further generate newer generations of offspring.%s of newer generations. 

 % Immigrants in dimension $k$ are generated from a homogenous Poisson process with intensity $\mu_k$, and they each independently generate a cluster of offsprings, which can further generate offsprings of newer generations. 

The branching structure, which is typically latent and unobservable, allows us to decouple the complex observed likelihood into factorizable terms and design simpler computational algorithms. The complete data log-likelihood, defined as the joint log-likelihood of the observed data and the branching structure $\mathbf{B}$, has the following form :
%
%For an event in dimension $j$, the number of offsprings generated in dimension $k$ follows a Poisson distribution with parameter $\alpha_{k,\ell}$, and the time between these events follows an exponential distribution with rate $\beta_{k,\ell}$. 
%
%{\color{red} Add here a bit more of an explanation about what the parameters mean in this context:  Immigrants generated from a homogenenous Poisson process with rate $\mu_k$, number of offspring of one event given by a Poisson distribution with parameter $\alpha_{k,\ell}$, and time to event given by an exponential distribution with rate parameter $\beta_{k,\ell}$.} 
%
%
%\begin{multline}\label{eq:complik}
%\mathcal{L}(\mathbf{X}, \mathbf{B} \mid \boldsymbol{\theta})=\sum_{k=1}^K\left|I_k\right| \log \mu_k+\sum_{j=1}^K \sum_{k=1}^K [\left|O_{k,\ell}\right|\log \left(\alpha_{k,\ell}\beta_{k,\ell}\right) \\-\beta_{k,\ell} \sum_{\substack{1 \leq i^{\prime}<i \leq N \\ t_{i'}^j<t_i^k \in \mathbf{X}}} B_{i i^{\prime}}\left(t_i^k-t_{i^{\prime}}^j\right)] -\sum_{k=1}^K \int_0^\intercal \lambda_k(s) ds
%\end{multline}
%
\begin{equation}
\label{eq:complik}
\begin{aligned}
\mathcal{L}(\mathbf{X}, \mathbf{B} \mid \boldsymbol{\alpha}, \boldsymbol{\beta}, \boldsymbol{\mu})&=\sum_{\ell=1}^K\left|I_\ell\right| \log \mu_\ell+\sum_{k=1}^K \sum_{\ell=1}^K\left[ \left|O _{k,\ell}\right|\left(\log \alpha _{k,\ell}+\log \beta _{k,\ell}\right)-\beta _{k,\ell} \sum_{\substack{j < i \\ d_{j} = k, d_{i} = l}} B_{ij}\left(t_i-t_{j}\right)\right] \\
& \;\;\;\;\;\;\;\;\;\;\;\;\;\;\;\;\;\;\;\;\;\;\;\;\;\;\;\;\;\;\;\;\;\;\;\;\;\;\;\;-\sum_{\ell=1}^K \mu_\ell T
-\sum_{k=1}^K \sum_{\ell=1}^K \alpha _{k,\ell}\left[n_k-\sum_{d_i=k} \exp \left(-\beta _{k,\ell}\left(T-t_i\right)\right)\right]  ,
\end{aligned}
\end{equation}
%
%\begin{multline}
%\label{eq:complik}
%\mathcal{L}(\mathbf{X}, \mathbf{B} \mid \boldsymbol{\theta})=\sum_{k=1}^K\left|I_k\right| \log \mu_k+\sum_{j=1}^K \sum_{k=1}^K \left|O_{k,\ell}\right|\log \left(\alpha_{k,\ell}\beta_{k,\ell}\right) 
%%
%- \sum_{j=1}^K \sum_{k=1}^K\beta_{k,\ell} \sum_{\substack{1 \leq i^{\prime}<i \leq N  t_{i'}^j<t_i^k \in \mathbf{X}}} B_{i i^{\prime}}\left(t_i^k-t_{i^{\prime}}^j\right) \\
%%
%-\sum_{k=1}^K \int_0^\intercal \lambda_k(s) ds
%\end{multline}
%
%
%
where $|I_\ell| = \sum_{\substack{1 \leq i \leq n \\ d_i = \ell}} B_{ii}$ is the number of immigrants for dimension $\ell$, and $|O_{k,\ell}| = \sum_{\substack{j < i \\ d_{j} = k, d_{i} = \ell}}B_{ij}$ is the number of descendants on dimension $k$ that are an offspring of an event on dimension $j$.

%%%%%%%%%%%%%%%%%%%%%%%%%%%%%%
\subsection{Approximation for the data likelihood}\label{sec:approx}

The expressions for the observed and complete data likelihood in \eqref{eq:obslik}  and \eqref{eq:complik} share the same term
\begin{align}
\Upsilon =	\sum_{\ell=1}^{K} \int_{0}^T \lambda_{\ell}(s) d s =\sum_{\ell=1}^K\mu_\ell T +  \sum_{k=1}^K\sum_{\ell=1}^K  \alpha_{k,\ell}\left[\sum_{d_i=k} \left(1 - \exp \left\{-\beta_{k,\ell}\left(T-t_i\right)\right\}\right)\right] .
\end{align}
\begin{figure*}[t]
	\centering
	\includegraphics[width = 0.7\linewidth]{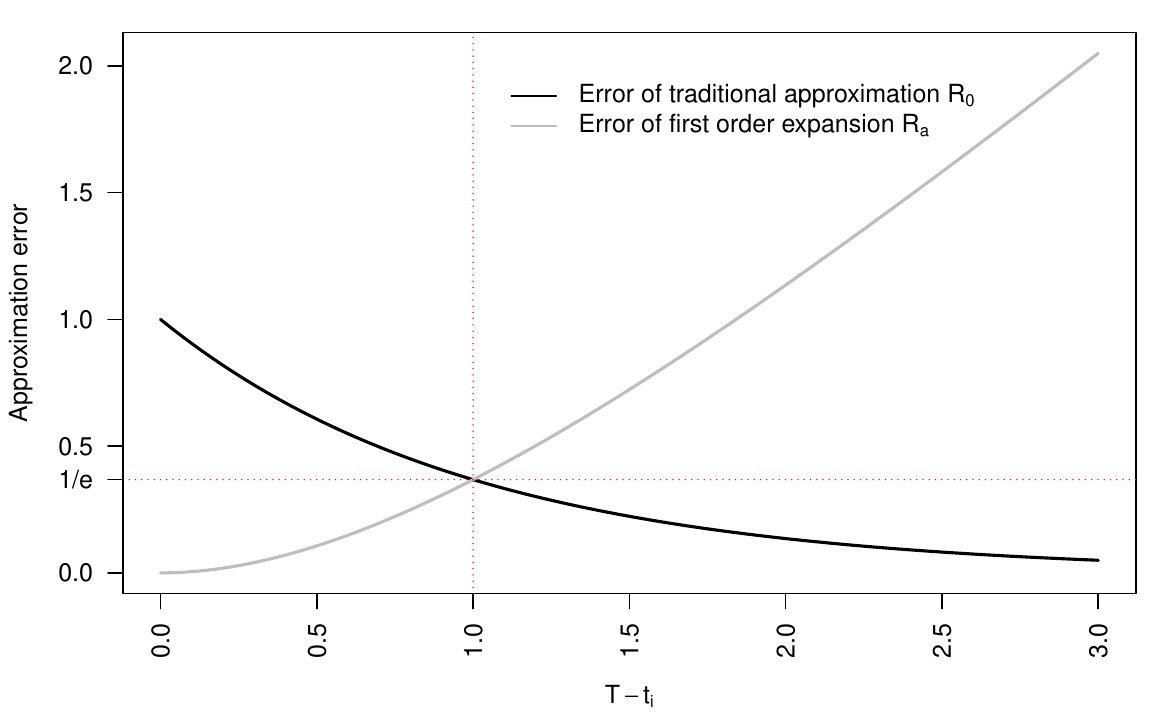}
	\caption{Trade off between approximation errors for the compensator for $\beta_{k,\ell} = 1$.\label{fig:approxerror}}
\end{figure*}

The integral $\int_{0}^T \lambda_{\ell}(s) d s$ is known as the compensator for the conditional density function $\lambda_{\ell}(t)$. The compensator guarantees that there are infinitely many `none events' between observations on the finite temporal region $[0,T]$ \citep{mei2017neural}.  The form of the compensator causes a number of computational challenges for designing scalable inference algorithms for MHP models (for a discussion see \citealp{Lewis11}, \citealp{schoenberg2013facilitated}, \citealp{holbrook2021scalable}). A common approach used to avoid these challenges is to use the approximation technique introduced in \cite{Lewis11}: 
\begin{equation}\label{eq:ass}
    \alpha_{k,\ell}\left[1 - \exp \left(-\beta_{k,\ell}\left(T-t_i\right)\right)\right] \approx \alpha_{k,\ell},
\end{equation}
for all $k, \ell = 1, \dots , K$. {\color{black} Clearly, this approximation is quite accurate for large values $T-t_i$.  Therefore, if the process is stationary and $T$ is ``large'' (ensuring that most observation are far away from $T$), the error $R_0 (T-t_i) = \exp\left\{-\beta_{k,\ell}(T-t_i)\right\}$ is negligible for most $i$.
%that most events are far away from the boundary and most exponential terms are therefore close to zero. The approximation is therefore most accurate for large datasets.% where the vast majority of the observations are far away from the boundary $T$.  For small datasets, this approximation can introduce edge effects, which can be specially problematic in the implementation of stochastic gradient methods. 
However, for observations that are ``close'' to the boundary $T$, the approximation is quite poor.  This will be an important consideration when designing algorithms for the MHP model that rely on subsamples because, in that setting, the equivalent quantity to $T$ will be small and the edge effects cannot necessarily be assumed to be negligible.

%An alternative motivation for the approximation in \eqref{eq:ass} is as a zero-order Taylor expansion of the exponential function. %, evaluated at $t \rightarrow \infty$.

To address this issue, we propose to use a different approximation for observations that are close to $T$. This alternative approximation is based on a first order Taylor expansion of the exponential function, so that 
%\begin{align}\label{eq:corr}
%\alpha_{k,\ell}\left[1 - \exp \left\{-\beta_{k,\ell}\left(T-t_i\right)\right\}\right] = \alpha_{k,\ell} \left(1 - \left[ 1 - \beta_{k,\ell}(T - t_i) + \sum_{n=2}^{\infty} \frac{ (-\beta_{k,\ell}(T-t_i))^n }{n!} \right] \right) \approx \alpha_{k,\ell} \beta_{k,\ell}(T - t_i)
%\end{align}
\begin{align}\label{eq:corr}
\alpha_{k,\ell}\left[1 - \exp \left\{-\beta_{k,\ell}\left(T-t_i\right)\right\}\right] \approx \alpha_{k,\ell} \left(1 - \left[ 1 - \beta_{k,\ell}(T - t_i)  \right] \right) = \alpha_{k,\ell} \beta_{k,\ell}(T - t_i)
\end{align}

The error associated with this approximation is $R_a (T-t_i) = \sum_{n=2}^{\infty} \frac{ (-\beta_{k,\ell}(T-t_i))^n }{n!} = \exp \left\{-\beta_{k,\ell}\left(T-t_i\right)\right\} - 1 + \beta_{k,\ell}(T-t_i)$.  Hence, as opposed to \eqref{eq:ass}, the approximation in \eqref{eq:corr} is accurate for small $T - t_i$.  Furthermore, $R_a$ is smaller than $R_0$ if and only if $T-t_i > 1/\beta_{k,\ell}$ (see Figure \ref{fig:approxerror}).  This suggests dividing $\mathbf{X}$ into two parts, based on whether the data points are observed within a predetermined threshold $(T-\delta, T]$ where $\delta = 1/\beta_{k,\ell}$, so that:
\begin{align}
\label{corr2}
	\alpha_{k,\ell}\left[\sum_{d_i=k} \left(1 - \exp \left\{-\beta_{k,\ell}\left(T-t_i\right)\right\}\right)\right] \approx   \alpha_{k,\ell} \left[ n_k - \sum_{0 \leq T-t_i < \delta, d_i = k } \left[1 - \beta_{k,\ell}(T - t_i) \right]  \right],	
\end{align}
for all $k, \ell = 1, \dots , K$. We call this the boundary-corrected approximation. 
One of its key advantages that it still allows us to exploit conjugacies in the definition of the model while providing a more accurate approximation for observations close to $T$. Please see Section \ref{sec:compmethods} for additional details. %We will apply both techniques to our inference algorithms and compare the model fitting results in section \ref{sec:Exp}.  
%In practice, the exponential terms in the compensators causes computational challenges and violates Assumption \ref{ass:approx}, thereby making the closed-form update for most inference algorithms infeasible. However, we note that $\tilde{t}_i \ll T$ for most $\tilde{t}_i \in \mathbf{X}$ and propose the following approximated compensator as the basis of developing scalable algorithms for the MHP model:

Note that, under our proposed approximation, the worst-case absolute error in the likelihood is bounded above by $n_{k,\ell}/e$.  Furthermore, the approximation error is guaranteed to always be smaller than that from the original approximation in \cite{Lewis11}.  Not only that, but it can also be shown that the average loss of information due to the approximation is negligible as $T\to \infty$ (see Appendix \ref{ap:theorem}).}

%%%%%%%%%%%%%%%%%%%%%%%%%%%%%%
\subsection{Prior distributions}
\label{sec:hyperpar}

Bayesian inference for the MHP requires that we specify priors for the unknown parameters  $(\boldsymbol{\alpha}, \boldsymbol{\beta}, \boldsymbol{\mu})$. For the baseline intensities we set
\begin{align*}
\mu_{\ell} \mid a_{\ell}, b_{\ell} &  \stackrel{i . i . d}{\sim} \operatorname{Gamma}\left(a_{\ell}, b_{\ell}\right), %k = 1,...,K ,
\end{align*} 
%b_{\ell} \mid c, d & \sim \operatorname{Gamma}(c, d) \\
which is conditionally conjugate given the branching structure $\mathbf{B}$.  Similarly, under the exponential decay functions we use
\begin{align*}
\alpha_{k,\ell} \mid e_{k,\ell}, f_{k,\ell} & \stackrel{i . i . d}{\sim} \operatorname{Gamma}\left(e_{k,\ell}, f_{k,\ell}\right),  \\
\beta_{k,\ell} \mid w_{k,\ell}, s_{k,\ell} & \stackrel{i . i . d}{\sim} \operatorname{Gamma} \left(w_{k,\ell}, s_{k,\ell}\right)
%f_{D} & \sim \operatorname{Gamma}\left(g_{D}, h_{D}\right), f_{O} \sim \operatorname{Gamma}\left(g_{O}, h_{O}\right) ,
\end{align*} 
which are also conditionally conjugate.  
%The Gamma prior distribution for the parameters preserves the conjugate structure and leads to efficient algorithm implementations. We also note that the immigrant events may trigger offspring events in all dimensions at different speeds, and there may be different patterns for generating offspring events in the same dimension (self-exciting behaviors) and different dimensions (mutually-exciting behaviors). We allow for the model to capture such heterogeneity by assigning different priors to the diagonal and off-diagonal elements in the $\boldsymbol{\alpha}$ matrix. %Finally, we let $\pi_{kj}$ be a sparsity hyperparameter, where we can control if events from dimension $j$ can trigger offspring events in dimension $k$ by choosing different values of $\pi_{kj}$: if $\pi_{kj} = 0$, the prior mean for $\alpha_{kj}$ shrink heavily towards zero compared to the case where $\pi_{kj} = 1$.

%%%%%%%%%%%%%%%%%%%%%%%%%%%%%%
\section{Computational methods}\label{sec:compmethods}

\subsection{Preliminaries}\label{prelim}

In this section, we describe three stochastic gradient algorithms for MHP models based on the EM, variational inference and an Markov chain Monte Carlo algorithm based on Langevin dynamics algorithm, respectively. Before delving into the details of each algorithm, we discuss three issues that are relevant to the design of all of them.

The first issue refers to how to define the subsamples used to compute the gradient at each iteration.  A common approach for regression models is to randomly select independent observations.  However, the temporal dependence in event sequences makes this approach inappropriate for MHP models. Instead, our subsamples consist of all observations contained in the random interval $[T_0, T_0 + \kappa T]$, where we uniformly sample $T_0$ on $[0, (1-\kappa)T]$ and $\kappa \in (0,1]$ corresponds to the relative size of the subsample. Similar strategies have been applied to developing stochastic gradient variational algorithms for hidden Markov models \citep{foti2014stochastic} and stochastic block models \citep{gopalan2012scalable}.

The second issue relates to the selection of the learning rate $\rho_r$ for the algorithms, which controls how fast the information from the stochastic gradient accumulates. It is well known (e.g., see \citealp{Robbins51}) that the following conditions lead to convergence towards a local optima:
\begin{align}
	\sum_{r=1}^{\infty} \rho_r &= \infty, & \sum_{r=1}^{\infty} \rho_r^2 &< \infty.
\end{align}
In the following analysis, we apply the commonly used update schedule for $\rho_r$, outlined in \cite{welling2011}:
\begin{align}\label{eq:stepsize}
	\rho_{r}=\rho_0(r+\tau_1)^{-\tau_2},
\end{align}
where $\rho_0$ is a common scaling factor, $\tau_2 \in (0.5,1]$ is the {forgetting rate} that controls the exponential decay rate, and $\tau_1 \geq 0$ is the {delay parameter} that downweights early iterations. In our numerical experiments, we investigate the impact of specific choices of  $\tau_1$ and $\tau_2$ on the results.

The third issue relates to the use of approximation techniques. We propose to use \eqref{corr2} to approximate the likelihood only for the stochastic gradient EM and variational inference algorithms since (conditional) conjugacy is important for developing these algorithms. In contrast, for stochastic gradient Langevin dynamics, we propose to use the exact likelihood in \eqref{eq:obslik}.  For simplicity, we will refer to the algorithms with approximation as their `boundary-corrected' versions, and we only show the update formula based on the `common approximation approach' in this section. Additionally, we want to point out that the exponential decay function that we are using allows us to update 
$\boldsymbol{\mu}$ and $\boldsymbol{\alpha}$ using the exact likelihood formula, and we will only consider the approximation when we update $\boldsymbol{\beta}$.

\subsection{Stochastic gradient EM algorithm for posterior mode finding}
\label{sec:SGEM}

The expectation-maximization (EM) algorithm \citep{dempster1977maximum} is an iterative maximization algorithm that is commonly used for latent variable models, especially in cases where knowledge of the latent variables simplifies the likelihood function. {For Bayesian models, it can be used for maximum \textit{a posteriori} probability estimation for the model parameters.} Let $\mathbf{X}$ be the observed dataset of size $N$, $\boldsymbol{\theta}$ be the set of model parameters to be estimated, and $\mathbf{B}$ be the set of latent branching structure variables, and denote $(\mathbf{X}, \mathbf{B})$ as the complete dataset. We further assume that the distribution of the complete dataset belongs to the following exponential family:  
\begin{align}
\label{ass1}
	l(\mathbf{X}, \mathbf{B} \mid \boldsymbol{\theta}) = A(\mathbf{X}, \mathbf{B})
	 \exp \left(\boldsymbol{\phi}(\boldsymbol{\theta})^\intercal \boldsymbol{s}(\mathbf{X}, \mathbf{B}) - \psi(\boldsymbol{\theta})  \right),
\end{align}
where $\boldsymbol{s}(\mathbf{X}, \mathbf{B})$ is the vector of sufficient statistics for the complete data model and $\boldsymbol{\phi}(\boldsymbol{\theta})$ is the canonical form of the vector of parameters and $\boldsymbol{\phi}(\cdot)$\ represents a one-to-one transformation. 

In the context of Bayesian models, the EM algorithm can be used to obtain the maximum a posteriori (MAP) estimate (e.g., see \citealp{logothetis1999}). Starting from an initial guess of the model parameters $\boldsymbol{\theta}^{(0)}$, the EM algorithm alternatively carries out the following two steps until convergence: 
\begin{itemize}
    \item In the `E-step', the algorithm estimates the ``marginal'' sufficient statistics based on the expected value $\hat{\boldsymbol{s}}^{(r)} :=  \mathrm{E}_{\mathbf{B} \mid \mathbf{X}, \boldsymbol{\theta}^{(r)}} \left[\boldsymbol{s}(\mathbf{X}, \mathbf{B})\right]$.%, conditioning on the data and the current parameter update $\boldsymbol{\theta}^{(r)}$.%${\color{red} $\boldsymbol{\theta}^{(r)} = (\boldsymbol{\mu}^{(r)}, \boldsymbol{\alpha}^{(r)}, \boldsymbol{\beta}^{(r)})$}.

    \item  In the `M-step', the algorithm updates the model parameter as the maximizer of the $Q$ function: 
\begin{align*}
	\boldsymbol{\theta}^{(r+1)} = \arg \min_{\boldsymbol{\theta}}\left[  
\boldsymbol{\phi}(\boldsymbol{\theta})^\intercal \hat{\boldsymbol{s}}^{(r)} + \log p( \boldsymbol{\theta})\right].
\end{align*}
where $p( \boldsymbol{\theta})$ denotes the prior on $ \boldsymbol{\theta}$.
\end{itemize}

Note that the expectation calculation in the E-step update requires a pass through the whole dataset.  As we discussed in the introduction, this can be challenging in very large datasets. The stochastic gradient EM (SGEM) algorithm \citep{Cappe09} addresses this challenge by approximating the marginal sufficient statistics with an estimate based on randomly sampled mini-batches. We let $\mathbf{X}^{(r)}$ denote a subsample of size $n$ (and respectively, we let $\mathbf{B}^{(r)}$ be the set of branching structure that corresponds to the selected subsample). For the stochastic E-step, the SGEM updates the estimated sufficient statistics $\hat{\boldsymbol{s}}^{(r+1)}$ as a linear combination of the previous update and a new estimate of the sufficient statistics based on the random subsample and the current model parameter:
\begin{align*}
    \hat{\boldsymbol{s}}^{(r+1)} =(1-\rho_r)\hat{\boldsymbol{s}}^{(r)} + \rho_r \kappa^{-1}\mathrm{E}_{\mathbf{B}^{(r+1)} \mid \mathbf{X}^{(r+1)}, \boldsymbol{\theta}^{(r)}}[\boldsymbol{s}(\mathbf{X}^{(r+1)}, \mathbf{B}^{(r+1)})].
\end{align*}
where $\rho_r$ is given in \eqref{eq:stepsize}.  Because of the way we select the subsamples, $\kappa^{-1}\mathrm{E}_{\mathbf{B}^{(r+1)} \mid \boldsymbol{X}^{(r+1)}, \boldsymbol{\theta}^{(r)}}[\boldsymbol{s}(\mathbf{X}^{(r+1)}, \mathbf{B}^{(r+1)})]$ is an unbiased estimate of the sufficient statistics of the model based on the whole dataset. In the following M-step, the SGEM algorithm maximizes the $Q$ function %computed on the basis of $\hat{S}^{(r+1)}$:
\begin{align*}
	\boldsymbol{\theta}^{(r+1)}=\arg \max _{\boldsymbol{\theta}}\left[\boldsymbol{\phi}(\boldsymbol{\theta})^\intercal\hat{\boldsymbol{s}}^{(r+1)}+\log p(\boldsymbol{\theta})\right].
\end{align*}
%Finally, we choose an update scheme for the \textit{stepsize parameter} $\rho_r$. 

%Alternatively, the EM algorithm can be viewed as an iterative approach to maximize a lower bound of the log-likelihood function \cite{Bishop2006}. However, we note that while the EM algorithm is a greedy algorithm in that the observed data likelihood will not decrease during each iteration, the model parameters may converge in local modes [citations!]. 

In the case of the MHP model with exponential excitation functions, $\boldsymbol{\theta}^{(r)} = (\boldsymbol{\mu}^{(r)}, \boldsymbol{\alpha}^{(r)}, \boldsymbol{\beta}^{(r)})$ and the expectation in the E-step is computed with respect to the probabilities
\begin{align}
\label{eq:p}
	p^{(r)}_{i,j}:= p\left(\mathbf{B}^{(r)}_{i,j}=1,\mathbf{B}^{(r)}_{i,-j}=0 \mid \boldsymbol{\mu}^{(r)}, \boldsymbol{\alpha}^{(r)}, \boldsymbol{\beta}^{(r)}, \mathbf{X}^{(r)}\right) \propto 
	\begin{cases}
	{\mu^{(r)}_{d_i} }~ & \text{if }  j = i,   \\ 
	{\alpha^{(r)}_{d_j,d_i}\beta^{(r)}_{d_j,d_i}\exp(-\beta^{(r)}_{d_j,d_i}(t_i-t_j)) } & \text{if }  j < i, \\
    0 & \text{if }  j > i.
	\end{cases}
\end{align}
for $i = 2, \dots, n$, the negative subindex stands for all other possible except the one, and $p_{1,1}^{(r)}:= 1$. Then, the vector of expected sufficient statistics of the complete data likelihood evaluated at iteration $r$  
$$
\left(s_{\mu ,\ell,1}^{(r)}, s_{\mu ,\ell,2}^{(r)}, s_{\alpha, k,\ell,1}^{(r)}, s_{\alpha ,k,\ell,2}^{(r)}, s_{\beta, k,\ell,1}^{(r)}, s_{\beta ,k,\ell,2}^{(r)}\right),
$$
is updated as
\begin{align*}
s_{\mu, \ell,1}^{(r+1)} &= (1-\rho_r)s_{\mu,\ell,1}^{(r)} + \rho_r \kappa^{-1} \sum_{d_i = \ell} p_{i,i}^{(r)},\\	
	s_{\mu ,\ell,2}^{(r+1)} &= T, \\
	s_{\alpha,k,\ell,1}^{(r+1)} &= (1-\rho_r)s_{\alpha, k,\ell,1}^{(r)} + \rho_r \kappa^{-1}\sum_{d_i=\ell} \sum_{\substack{d_{j}=k \\ j<i}} p_{i,j}^{(r)},\\ 
	s_{\alpha,k,\ell,2}^{(r+1)} &= (1-\rho_r)s_{\alpha,k,\ell,1}^{(r)} + \kappa^{-1}\left(n_j^{(r)} - \sum_{d_j = k}\exp \left( - \beta_{k,l}^{(r)} \left(\kappa T - t_j \right) \right)\right) ,\\
    s_{\beta, k,\ell,1}^{(r+1)} &=
        (1-\rho_r)s_{\beta ,k,\ell,1}^{(r)} + \rho_r \kappa^{-1} \sum_{d_i=l} \sum_{\substack{d_{j}=k \\ j<i}}p_{i,j}^{(r)} \\
    s_{\beta, k,\ell,2}^{(r+1)} &= (1-\rho_r)s_{\beta ,k,\ell,2}^{(r)} + \rho_r\kappa^{-1}\sum_{d_i=l} \sum_{\substack{d_{j}=k \\ j<i}} p_{i,j}^{(r)}\left(t_{i}-t_{j}\right),
\end{align*}
\begin{comment}
    \begin{cases}
        (1-\rho_r)s_{\beta ,k,\ell,1}^{(r)} + \rho_r \kappa^{-1} \sum_{d_i=l} \sum_{\substack{d_{j}=k \\ j<i}}p_{i,j}^{(r)} ~ & \text{(uncorrected)} \\
        (1-\rho_r)s_{\beta ,k,\ell,1}^{(r)} + \rho_r \kappa^{-1}\left( \sum_{d_i=l} \sum_{\substack{d_{j}=k \\ j<i}}p_{i,j}^{(r)} + \alpha_{k,\ell}^{(r)}\sum_{\substack{ d_i=k \\ 0 \leq T-t_i<\delta}}(\kappa T - t_i) \right) ~ & \text{(corrected)}
    \end{cases}
\end{comment}
where $n_j^{(r)}$ denotes the number of events on dimension $j$ in $\mathbf{X}^{(r)}$. Finally, in the M-step, the value of the parameters is updated as:
\begin{align*}
{\alpha}_{k,\ell}^{(r+1)} &= \frac{s_{\alpha ,k,\ell,1}^{(r+1)} +e_{ k,\ell}-1}{s_{\alpha ,k,\ell, 2}^{(r+1)} + f _{k,\ell}}, &
{\beta} _{k,\ell}^{(r+1)} & = \frac{s_{\beta, k,\ell,1}^{(r+1)} + w_{k,\ell}-1}{s_{\beta ,k,\ell,2}^{(r+1)}+s_{k,\ell}} , &
{\mu}_{\ell}^{(r)} &=\frac{s_{\mu, \ell,1}^{(r+1)}+a_{\ell}-1}{s_{\mu ,\ell,2}^{(r+1)}+b_{\ell}}.
\end{align*}
We repeat the steps above until the convergence criterion is reached.

\subsection{Stochastic Gradient Variational Inference}\label{sec:SGVI}

\begin{comment}
by minimizing the Kullback-Leibler (KL) divergence between the approximation and the posterior, i.e., by setting
\end{comment}
Variational inference \citep{wainwright2008graphical} is an approximate inference method that replaces the posterior distribution with an approximation that belongs to a tractable class.  More specifically, the variational approximation $q_{\boldsymbol{\eta}}(\boldsymbol{\theta},\mathbf{B})$, $\boldsymbol{\eta} \in H$ to the posterior distribution $p(\boldsymbol{\theta}, \mathbf{B}  \mid \mathbf{X})$ is obtained through maximizing the evidence lower bound (ELBO), which is equivalent to setting 
\begin{align}\label{eq:KLvar}
\boldsymbol{\eta} = \arg\max_{\tilde{\boldsymbol{\eta}} \in H}  \operatorname{E}_{q_{\tilde{\boldsymbol{\eta}}}} \log \left\{ \frac{p(\boldsymbol{\theta}, \mathbf{B} , \mathbf{X})}{q_{\tilde{\boldsymbol{\eta}}}(\boldsymbol{\theta},\mathbf{B})} \right\} .
\end{align}

The class of variational approximations most used in practice is the class of mean-field approximations \citep{bishop2006pattern}, where  model parameters are taken to be independent from each other under the variational distribution, i.e., $q_{\boldsymbol{\eta}}(\boldsymbol{\theta},\mathbf{B}) = \prod_{j}q_{\boldsymbol{\eta}_{\theta_j}}(\theta_j)\prod_{i}q_{\boldsymbol{\eta}_{\mathbf{B}_i}}(\mathbf{B}_i)$.  If both the full conditional posterior distributions and the corresponding variational distribution belong to the same exponential family, e.g., if
\begin{align*}
  p\left(\theta_j \mid \boldsymbol{\theta}_{-j},  \mathbf{B}, \mathbf{X}\right) & =A\left(\theta_j\right) \exp \left\{\theta_j s_j\left(\boldsymbol{\theta}_{-j}, \mathbf{B}, \mathbf{X}\right)-\psi\left(\boldsymbol{\theta}_{-j}, \mathbf{X}\right)\right\} , &
  q_{\boldsymbol{\eta}_{\theta_j}}(\theta_j) &= A(\theta_j) \exp \left\{  \theta_j s_i\left(\boldsymbol{\eta}_{\boldsymbol{\theta}_j}\right)  - \psi\left(\boldsymbol{\eta}_{\boldsymbol{\theta}_j}\right) \right \} ,
\end{align*}
%and
%\begin{align*}
%  p\left(B_i \mid \mathbf{B}_{-i}, \boldsymbol{\theta},  \mathbf{B}, \mathbf{X}\right) & =\tilde{A}\left(B_i\right) \exp \left\{ B_i \tilde{s}_i \left(\mathbf{B}_{-i},\boldsymbol{\theta},  \mathbf{X}\right)-\tilde{\psi}\left(\mathbf{B}_{-i}, \mathbf{X}\right)\right\} , &
  %
%  q_{\boldsymbol{\eta}_{B_i}}(B_i) &= \tilde{A}(B_i) \exp \left\{  B_i s_i\left(\boldsymbol{\eta}_{B_i}\right)  - \psi\left(\boldsymbol{\eta}_{B_i}\right) \right\} ,
%\end{align*}
\cite{Blei06} showed that the coordinate ascent algorithm for the mean-field variational inference updates the variational parameters by setting $\boldsymbol{\eta}_{\theta_j}^{(r+1)} = \operatorname{E}_{q_{\boldsymbol{\eta^{(r)}}}} \left[s_{j}\left(\boldsymbol{\theta}_{-j}, \mathbf{B}, \mathbf{X}\right)\right]$.  A similar result applies to the updates of the variational parameters $\boldsymbol{\eta}_{B_i}$.% and $\boldsymbol{\eta}_{B_i}^{(r+1)} = \operatorname{E}_{q_{\boldsymbol{\eta^{(r)}}}} \left[\tilde{s}_{i}\left(\mathbf{B}_{-i}, \boldsymbol{\theta},  \mathbf{X}\right)\right]$.

Stochastic gradient variational inference (SGVI) \citep{Hoffman13} is a variant of variational inference that replaces the gradient computed over the whole sample with the one calculated over a random subsample $\mathbf{X}^{(r)}$ of size $n$ selected during iteration $r$.  Under conjugacy, SGVI then updates the vector $\boldsymbol{\eta}_\mathbf{B}$ in iteration $r$ by setting
\begin{align*}
	\eta^{(r+1)}_{\mathbf{B}_i^{(r)}} = \operatorname{E}_{q_{\eta^{(r)}}}\left[\tilde{s}_i\left( \mathbf{B}^{(r)}_{-i}, \boldsymbol{\theta},  \mathbf{X}^{(r)} \right)\right], 
\end{align*}
where $\tilde{s}_i\left( \mathbf{B}^{(r)}_{-i}, \boldsymbol{\theta},  \mathbf{X}^{(r)} \right)$ is the sufficient statistics associated with the block $\mathbf{B}_i$, and $\boldsymbol{\eta}_{\boldsymbol{\theta}}$ through the recursion 
\begin{align*}
	\eta_{\boldsymbol{\theta}_j}^{(r+1)} = (1-\rho_r)\eta_{\boldsymbol{\theta}_j}^{(r)} + \rho_r\hat{\eta}^{(r+1)}_{\boldsymbol{\theta}_j} ,
\end{align*}
where $\hat{\eta}^{(r+1)}_{\boldsymbol{\theta}_j} = \operatorname{E}_{q_{\eta^{(r+1)}}}\left[ s_j(\boldsymbol{\theta}_{-j}, \mathbf{B}^{(r)}, \mathbf{X}^{(r)}) \right]$.  In the specific case of the MHP with exponential excitation functions we have $\boldsymbol{\theta} = (\boldsymbol{\mu}, \boldsymbol{\alpha}, \boldsymbol{\beta})$, $\boldsymbol{\eta} = \left(\boldsymbol{\eta}_{\boldsymbol{\alpha}}, \boldsymbol{\eta}_{\boldsymbol{\beta}}, \boldsymbol{\eta}_{\boldsymbol{\mu}}, \boldsymbol{\eta}_{\boldsymbol{B}}\right)$ and 
\begin{align*}
	q_{\boldsymbol{\eta}}(\boldsymbol{\alpha}, \boldsymbol{\beta}, \boldsymbol{\mu}, \mathbf{B}) &= \prod_{i=1}^N q_{\boldsymbol{\eta}_{\mathbf{B}_i}}(\mathbf{B}_i) \prod_{k=1}^K q_{\boldsymbol{\eta}_{\boldsymbol{\mu_k}}}(\mu_k) \prod_{j=1}^K \prod_{k=1}^K q_{\boldsymbol{\eta}_{\boldsymbol{\alpha_{k,\ell}}}}(\alpha_{k,\ell}) q_{\boldsymbol{\eta}_{\boldsymbol{\beta_{k,\ell}}}}(\beta_{k,\ell})  ,
\end{align*}
where
%$q_{\boldsymbol{\eta}_{\mathbf{B}_i}}$, $q_{\boldsymbol{\eta}_{\boldsymbol{\mu_k}}}$, $q_{\boldsymbol{\eta}_{\boldsymbol{\alpha_{k,\ell}}}}$ and $q_{\boldsymbol{\eta}_{\boldsymbol{\beta_{k,\ell}}}}$ take the form 
%
%where, under such distribution, the parameters $\boldsymbol{\alpha}, \boldsymbol{\beta}, \boldsymbol{\mu}$ are Gamma distributed and $\mathbf{B}_i$ is a binary vector, similarly as defined in Section \ref{sec:SGEMhp}:
%\begin{align*}
%	\alpha_{k,\ell} &\stackrel{i.i.d}{\sim} \operatorname{Gamma}(\eta_{\alpha ,k,\ell,1}, \eta_{\alpha ,k,\ell,2}), &
%	\beta_{k,\ell} &\stackrel{i.i.d}{\sim} \operatorname{Gamma}(\eta_{\beta ,k,\ell,1}, \eta_{\beta ,k,\ell,2}), &
%	\mu_{k} &\stackrel{i.i.d}{\sim} \operatorname{Gamma}(\eta_{\mu ,k,1}, \eta_{\mu ,k,2}), 
%\end{align*}
%
$\alpha_{k,\ell} \sim %\stackrel{i.i.d}{\sim} 
\operatorname{Gamma}(\eta_{\alpha ,k,\ell,1}, \eta_{\alpha ,k,\ell,2})$, $\beta_{k,\ell} \sim %\stackrel{i.i.d}{\sim} 
\operatorname{Gamma}(\eta_{\beta ,k,\ell,1}, \eta_{\beta ,k,\ell,2})$, $\mu_{\ell} 
 \sim %\stackrel{i.i.d}{\sim} 
\operatorname{Gamma}(\eta_{\mu ,\ell,1}, \eta_{\mu ,\ell,2})$, $\mathbf{B}_{i}$ denotes the $i$-th row of the matrix $\mathbf{B}$, and $\mathbf{B}_i$ follows a categorical distribution with parameter $\boldsymbol{\eta}_{\mathbf{B}_{i}}$. 
Hence, each iteration of the SGVI algorithm starts by updating the variational parameter for the local branching structure through the following formula:
\begin{align*}
\eta_{B^{(r)}_{i j}} &\propto \begin{cases}
\exp \left\{ \psi\left( \eta^{(r)}_{\mu, d_i ,1}\right) - \log \left(\eta^{(r)}_{\mu ,d_i, 2}\right)\right\} & j = i\\
\exp \left\{\Psi_{ij} - \log \left(\eta^{(r)}_{\alpha, d_j, d_i ,2} \right) - \log \left(\eta^{(r)}_{\beta, d_j, d_i ,2} \right) \right\} & j < i,  \\
0 & j > i, \\
\end{cases} 
\end{align*}
where $\Psi_{ij} = \psi\left(\eta^{(r)}_{\alpha, d_j,d_i,1}\right)+\psi\left(\eta^{(r)}_{\beta ,d_j,d_i,1}\right)-\frac{\eta^{(r)}_{\beta, d_j,d_i,1}}{\eta^{(r)}_{\beta ,d_j,d_i,2}}\left(t^{(r)}_{i}-t^{(r)}_{j}\right)$. In this expression, $\psi(x) = \frac{\mathrm{d}}{\mathrm{d} x} \ln \Gamma(x)$ denotes the digamma function, and $(t_i^{(r)},t_j^{(r)})$ represents the $i$-th and $j$-th event in $\mathbf{X}^{(r)}$. 
Then, we update the rest of the variational parameters as:
\begin{align*}
	{\eta}_{\alpha_{k,\ell,1}}^{(r+1)} &= (1-\rho_r){\eta}_{\alpha_{k,\ell,1}}^{(r)} + \rho_r\left(\kappa^{-1} \sum_{d_i = \ell} \sum_{\substack{d_{j}=k \\ j <i }}\eta_{B^{(r)}_{ij}}+e_{ k,\ell}\right), \\ 
	{\eta}_{\alpha_{k,\ell,2}}^{(r+1)} &= (1-\rho_r){\eta}_{\alpha_{k,\ell,2}}^{(r)} + \rho_r\left(\kappa^{-1} \left(n_{k}^{(r)} - \sum_{d_j = k} \left(1 + \frac{\kappa T-t_j}{\eta_{\beta_{k, \ell, 2}^{(r+1)}}}\right)^{-\eta_{\beta_{k,l,1}}^{(r+1)}}\right)+f_{k,\ell}\right),\\
    {\eta}_{\beta_{k,\ell,1}}^{(r+1)} &= (1-\rho_r){\eta}_{\beta_{k,\ell,1}}^{(r)} + \rho_r \left(\kappa^{-1}\sum_{d_i = \ell} \sum_{\substack{d_{j}=k \\ j <i }}\eta_{B^{(r)}_{ij}}+r _{k,\ell}\right), \\ 
    {\eta}_{\beta_{k,\ell,2}}^{(r+1)} &= (1-\rho_r){\eta}_{\beta_{k,\ell,2}}^{(r)} + 
      \rho_r \left(\kappa^{-1}\sum_{d_i = \ell} \sum_{\substack{d_{j}=k \\ j <i }}\eta_{B^{(r)}_{ij}}(t_i^{(r)} - t_{j}^{(r)})+s _{k,\ell}\right),\\
      {\eta}_{\mu_{\ell,1}}^{(r+1)} &= (1-\rho_r){\eta}_{\mu_{\ell,1}}^{(r)} + \rho_r\left(\kappa^{-1}\sum_{d_i=\ell} \eta_{B^{(r)}_{ii}}+a_{\ell}\right),\\ 
	{\eta}_{\mu_{\ell,2}}^{(r+1)} &= T + b_{\ell}.
\end{align*}
These updates are repeated until convergence.

\subsection{Stochastic Gradient Langevin Dynamics}\label{sec:sgld}

%
\begin{comment}
%%%%%%%%%%%%%%%%%%%%%%%%%%%%%%
focus on methods based on observed data likelihood rather than the complete data likelihood, instead we consider methods inspired by Langevin dynamics. ... what LD is ... MCMC based on LD led to biased random walk proposals which rae more efficient than traditional symmetric rw proposals. Mention LD use constant stepsize for discretization .
- text on what LD 
- U def, next proposal .
- acceptance 
IN SGLD, replace the gradient by the unbiased estiamte of the gradient based on the subsample, use a decreasing stepsize leading to updates of the form 
%%%%%%%%%%%%%%%%%%%%%%%%%%%%%%

The stochastic gradient MCMC methods (SGMCMC) are scalable MCMC approaches that generates samples from the desired posterior distribution by calculating likelihood gradients on subsamples, thereby avoiding the need to pass through the whole dataset for each iteration. Most SGMCMC methods are based on the discretization of continuous diffusions processes that admit the target posterior distribution as the stationary distribution. 
\end{comment} 
%

Unlike the previous two sections, here we focus on inference methods that are based on the observed data likelihood  \eqref{eq:obslik} instead of the complete data likelihood \eqref{eq:complik}. Specifically, we consider simulation methods that rely on Langevin dynamics (LD)  \citep{Neal2011}, a class of MCMC methods that are based on the discretization of a continuous-time stochastic process whose equilibrium distribution is the desired posterior distribution. Compared to simple random walk MCMC algorithms, LD algorithms explore the parameter space much more efficiently because they use information about the gradient of the likelihood to guide the direction of the random walk. In particular, LD methods proposes new values for the parameter according to 
\begin{align}\label{eq:LDupdate}
\boldsymbol{\theta}^{*} = \boldsymbol{\theta}^{(r)} - \frac{\rho}{2} \left. \nabla_{\boldsymbol{\theta}} U\left( \boldsymbol{\theta} \mid \mathbf{X} \right) \right|_{\theta = \boldsymbol{\theta}^{(r)}} + \sqrt{\rho} \epsilon_{r+1},    
\end{align}
where $\rho$ is the step size used to discretize the Langevin diffusion, $U\left( \boldsymbol{\theta} \mid \mathbf{X} \right) = - \log p(\mathbf{X} \mid \boldsymbol{\theta}) - \log p(\boldsymbol{\theta})$ is the negative logarithm of the unnormalized posterior of interest, and  $\epsilon_{r+1}$ is drawn from a standard multivariate normal distribution.  If no discretization of the Langevin diffusion was involved, then this proposed valued would come from the correct stationary distribution.  However, the introduction of the discretization means that a correction is required. Hence, values proposed according to \eqref{eq:LDupdate} are accepted with probability
\begin{align}\label{eq:accpetLD}
\min \left\{ 1, \frac{ \exp \left\{ - U(\boldsymbol{\theta}^{*} \mid \mathbf{X}) \right\}}{ \exp \left\{ - U(\boldsymbol{\theta}^{(r)} \mid \mathbf{X})\right\}} \right\}.    
%\min \left\{ 1, \frac{ p(\boldsymbol{\theta}^{*} \mid \mathbf{X})}{ p(\boldsymbol{\theta}^{(r)} \mid \mathbf{X})} \right\}.    
\end{align}
If accepted, then $\boldsymbol{\theta}^{(r+1)} = \boldsymbol{\theta}^{*}$.  Otherwise,  $\boldsymbol{\theta}^{(r+1)}=\boldsymbol{\theta}^{(r)}$.

The stochastic gradient Langevin Dynamics (SGLD) algorithm \citep{welling2011, chen2014stochastic} replaces the likelihood computed over the whole sample with (an appropriately rescaled) likelihood evaluated on a random subsample  $\mathbf{X}^{(r)}$.  SGLD also uses a decreasing stepsize $\rho_r$ to construct the discretization of the Langevin diffusion in step $r$ of the algorithm and ignores the correction step in \eqref{eq:accpetLD}.  This leads to updates of the form 
\begin{align}
\label{eq:5}
    \boldsymbol{\theta}^{(r+1)} = \boldsymbol{\theta}^{(r)} - \frac{\rho_r}{2} \left. \nabla_{\boldsymbol{\theta}} \tilde{U}(\boldsymbol{\theta} \mid \mathbf{X}^{(r)}) \right|_{\theta = \boldsymbol{\theta}^{(r)}} + \sqrt{\rho_r} \epsilon_{r+1},  
\end{align}
where $\tilde{U}(\boldsymbol{\theta} \mid \mathbf{X}^{(r)}) = \kappa^{-1} \log p\left(\mathbf{X}^{(r)} \mid \boldsymbol{\theta}\right) + \log p\left( \boldsymbol{\theta} \right)$.

In the case of the MHP model with exponential excitation functions, we perform a logarithmic transformation on the model parameters before implementing the SGLD, so that $\boldsymbol{\xi}_{\boldsymbol{\alpha}} = \log \boldsymbol{\alpha}$, $\boldsymbol{\xi}_{\boldsymbol{\beta}} = \log \boldsymbol{\beta}$ and $\boldsymbol{\xi}_{\boldsymbol{\mu}} = \log \boldsymbol{\mu}$.  Then, the gradients become:
\begin{equation*}
\begin{aligned}
    \nabla_{\xi_{\alpha_{k,\ell}}}^{(r)} U\left(\boldsymbol{\xi} \right)&=-\sum_{d_i =\ell} \frac{\alpha_{k,\ell}^{(r)}\beta_{k,\ell}^{(r)} \sum_{\substack{d_j = k , j < i}} \exp  \left(-\beta_{k,\ell}^{(r)}\left(t^{(r)}_{i}-t^{(r)}_{j}\right)\right)}{\mu_{\ell}^{(r)}+\alpha_{k,\ell}^{(r)}\beta_{k,\ell}^{(r)} \sum_{\substack{d_j = k , j < i}} \exp  \left(-\beta_{k,\ell}^{(r)}\left(t^{(r)}_{i}-t^{(r)}_{j}\right)\right)} \\ &+\alpha_{k,\ell}^{(r)}\left(n_k^{(r)} - \sum_{d_j=k} \exp \left(-\beta_{k, l}^{(r)}\left(\kappa T-t_j\right)\right) + f_{k,\ell}\right) - e_{k,\ell}, \\
    \nabla_{\xi_{\beta_{k,\ell}}}^{(r)} U\left(\boldsymbol{\xi} \right)&=-\sum_{d_i =\ell} \frac{\alpha_{k,\ell}^{(r)}\beta_{k,\ell}^{(r)} \sum_{d_j = k, j<i} \left(1-\beta_{k, \ell}^{(r)}\left(t^{(r)}_{i}-t^{(r)}_{j}\right)\right)\exp  \left(-\beta_{k,\ell}^{(r)}\left(t^{(r)}_{i}-t^{(r)}_{j}\right)\right)}{\mu_{\ell}^{(r)}+\alpha_{k,\ell}^{(r)}\beta_{k,\ell}^{(r)} \sum_{d_j = k, j<i} \exp  \left(-\beta_{k,\ell}^{(r)}\left(t^{(r)}_{i}-t^{(r)}_{j}\right)\right)} \\ &+\sum_{d_j = k} \alpha_{k,l}^{(r)}(\kappa T - t_j)\exp \left(-\beta_{k, l}^{(r)}\left(\kappa T-t_j\right)\right)   -r_{k,\ell}+s_{k,\ell} \beta^{(r)}_{k,\ell}, \\
    \nabla_{\xi_{\mu_{k,\ell}}}^{(r)} U\left(\boldsymbol{\xi} \right)&=-\sum_{d_i=\ell} \frac{\mu_{\ell}^{(r)}}{\mu_{\ell}^{(r)}+\alpha_{k, \ell}^{(r)} \beta_{k, \ell}^{(r)} \sum_{\substack{d_j=k ,j<i}} \exp \left(-\beta_{k, l}^{(r)}\left(t_i^{(r)}-t_j^{(r)}\right)\right)}+\mu^{(r)}_{\ell} (b_{\ell}+\kappa T)-a_{\ell}.
\end{aligned}
\end{equation*}
Note that SGLD does not require approximating the observed data likelihood.

\section{Simulation studies}
\label{sec:Exp}

In this section, we conduct a series of simulations to understand the performance of the algorithms with and without time budget constraints.  Compared with small-scale learning problems, large-scale problems are subject to a qualitatively different tradeoffs involving the computational complexity of the underlying algorithm \citep{bottou2007tradeoffs}, making evaluation under time constraints key. We also investigate the model fitting performance of all algorithms under different subsampling ratios.

\subsection{Experimental setting}
\label{sec:syn}

\paragraph{Data generation mechanism.} 

{\color{black} For most of our experiments,} data is generated from the multivariate Hawkes process model presented in section \ref{sec:mhp} with $K=3$ dimensions and the following parameter settings:
\begin{align*}
	\boldsymbol{\alpha} &= \begin{bmatrix}
		0.3 & 0.3 & 0.3 \\
		0.3 & 0.3 & 0.3 \\
		0.3 & 0.3 & 0.3 
	\end{bmatrix},  & 
	\boldsymbol{\beta} & = \begin{bmatrix}
		4 & 4 & 4 \\
		4 & 4 & 4 \\
		4 & 4 & 4 
	\end{bmatrix},  & 
	\boldsymbol{\boldsymbol{\mu}} = & \begin{bmatrix}
 0.5 \\ 0.5 \\ 0.5
 \end{bmatrix}.
\end{align*}

{\color{black} In addition, for our last set of experiments, we also consider an alternative data generation mechanism for a 10-dimensional Hawkes processes with varying degrees of sparsity on the matrix $\boldsymbol{\alpha}$. More specifically, we let have $\mu_\ell = 0.1, \beta_{k, \ell} = 4$  for all $k, \ell = 1, \dots, 10$ and  consider three scenarios (`high', `medium' and `low' sparsity), under which 90\%, 80\% and 50\% of the off-diagonal elements of $\boldsymbol{\alpha}$ are set to zero, indicating no interaction between the dimension pairs. The remaining off-diagonal elements of $\boldsymbol{\alpha}$ are set to the value 0.1, and the diagonal elements to the value 0.4.}

\paragraph{Algorithms to be compared.}

We compare the performances of SGLD, versions of SGEM, SGVI that use the standard approximation of \cite{Lewis11}, and the boundary-corrected versions of SGEM and SGVI (SGEM-c and SGVI-c). Also, as a `gold-standard' that does not involve subsampling, we implemented full MCMC and its boundary-corrected version (MCMC-c). {\color{black} Additionally, we benchmark our methods against two frequentist computational methods for MHP: the nonparametric estimation based on EM algorithm and piecewise basis kernels (EM-BK, see \citealp{zhou2013}) and the maximum likelihood estimation for multidimensional exponential Hawkes process with both excitation and inhibition effects (MLE-I, see \citealp{bonnet2023inference}).}

\paragraph{Parameters.} For the model hyperparameters from Section \ref{sec:hyperpar}, we let $a_{\ell} = 2, b_{\ell} = 4, e_{k,\ell} = 2, f_{k,\ell} = 4, r_{k,\ell} = 2, s_{k,\ell} = 0.5$ for $k, \ell = 1, \dots, K$. We simulate $K_d = 50$ datasets for $T = 1000$. For every dataset, we start all algorithms at $16$ different initial points to minimize the risk of convergence to a local optimum. For the tuning hyperparameters in stochastic optimization algorithms, we consider several subsampling ratios of $\kappa = \{0.01, 0.05, 0.1, 0.2, 0.3, 0.4\}$ and let $\tau = 1, \kappa = 0.51$. For SGEM and SGVI, we let $\rho_0 = 0.02$, and for SGLD we let $\rho_0 = \frac{0.1}{T\kappa}$. We chose $\delta = 0.25$ as the threshold for boundary-corrected methods.

\paragraph{Performance Metrics for Model Fitting.} 

We consider the observed data likelihood defined in \eqref{eq:obslik} as a measure for model fitting. Denote by $\text{ODL}_{d,\iota}$ the observed data likelihood calculated based on dataset $d$ and initial point $\iota$, we define $\text{BODL}_{d} = \max_{1 \leq \iota \leq 16} \text{ODL}_{d,\iota}$ as the {best-observed data likelihood} (BODL), as a basis for evaluating model performance.  Finally, in order to compare model-fitting performance under different subsampling ratios and different datasets, we propose the following relative best-observed data likelihood (RBODL):
$$
	\text{RBODL}_{d,\kappa_1, \kappa_2} = \frac{\text{BODL}_{d, \kappa_1}}{\text{BODL}_{d, \kappa_2}}
$$
where $\text{BODL}_{d, \kappa_1}, \text{BODL}_{d, \kappa_2}$ are the best-observed data likelihoods on dataset $d$ under subsampling ratio $\kappa_1$ and $\kappa_2$. Additionally, we refer to $\kappa_2$ as the reference subsampling ratio for $\text{RBODL}_{d,\kappa_1, \kappa_2}$.  The RBODL is fairly easy to interpret, in that $\text{RBODL}_{d,\kappa_1, \kappa_2} > 1$ indicates a superior empirical performance of subsampling ratio $\kappa_1$ compared to $\kappa_2$ and vice versa. %Also, RBODL allows us to study how the subsampling ratio affects model fitting performance for different datasets on the same basis, as different datasets may have very different BODLs due to sampling variability. 

\paragraph{Performance Metrics for Estimation Accuracy.}

We consider performance metrics for both point and uncertainty estimations. 
%In the following expressions, the parameters $\hat{\alpha}_{..}, \hat{\beta}_{..}, \hat{\mu}_{.}$ represent the estimated parameters from the respective three different methods. For SGVI, they are the means of the variational distribution based on the final iteration of the variational parameters from the coordinate ascent algorithm, for SGEM, they are the final iterates from the algorithm, for SGLD, they are the posterior means of the MCMC samples after burn-in.
To evaluate estimation accuracy of the model parameters, we rely on the averaged root mean integrated squared error (RMISE) for $\boldsymbol{\alpha}, \boldsymbol{\beta}$, and use mean absolute error (MAE) for $\boldsymbol{\mu}$ on the log scale: 

\begin{align*}
	\text{RMISE}(\boldsymbol{\alpha}, \boldsymbol{\beta}) :&= \frac{1}{K^2} \sum_{j=1}^K \sum_{k=1}^K \sqrt{\int_0^{+\infty}\left(\phi^{\text{true}}_{j, k}(x)-\hat{\phi}_{j, k}(x)\right)^2 \mathrm{~d} x}, \\
	\text{MAE}(\boldsymbol{\mu}) :&= \frac{1}{K} \sum_{j=1}^K |\log(\mu^{\text{true}}_k) - \log(\hat{\mu}_k)|.
\end{align*}
where $\hat{\mu}_k$ is the point estimator of $\mu_k$ (the posterior mode for the stochastic gradient EM, the posterior mean under the variational approximation for the stochastic gradient variational method, and the posterior mean of the samples after burn-in for the stochastic gradient Langevin dynamics), and $\hat{\phi}_{k,\ell}(x)$ is obtained by plugging in the point estimators for $\alpha_{k,\ell}$ and $\beta_{k,\ell}$ into the exponential decay function. The RMISE is a commonly used metric for nonparametric triggering kernel estimation for MHP models \citep{zhou2020} and collectively evaluates the estimation performance for all model parameters.

We also evaluate the uncertainty estimates generated by the SGVI, SGLD, SGVI-c and SGLD-c models (SGEM provides point estimators, but does not directly provide estimates of the posterior variance).  To do so, we consider the interval score (IS) \citep{gneiting2007strictly} for 95\% credible intervals, which jointly evaluates the credible interval width and its coverage rate. We also separately compute the average coverage rate (ACR), defined as the proportion of correct coverages out of $2K^2 + K$ model parameters and the average credible interval length (AIW) as references.

\subsection{Simulation results}

\begin{table*}[t]
\centering
\begin{tabular}{@{}ccccccc@{}}
\toprule
methods               & running time & 0.05         & 0.1          & 0.2          & 0.3          & 0.4          \\ \midrule
\multirow{3}{*}{SGEM} & 1 min        & \textbf{1.003} (0.003)&1.002 (0.004)&1.000 (0.003)&0.999 (0.004)&0.997 (0.004) \\
                      & 3 min        & \textbf{1.004} (0.005)&1.003 (0.005)&1.002 (0.005)&1.001 (0.006)&1.000 (0.005) \\
                      & 5 min        & 1.003 (0.005) & \textbf{1.004} (0.006)&1.002 (0.006)&1.002 (0.006)&1.001 (0.006) \\ \cmidrule(l){3-7} 
\multirow{3}{*}{SGEM-c} & 1 min        & \textbf{1.003} (0.007)&1.002 (0.006)&1.000 (0.006)&0.999 (0.006)&0.998 (0.006 \\
                      & 3 min        & \textbf{1.003} (0.005)&1.003 (0.006)&1.002 (0.006)&1.001 (0.006)&1.000 (0.006) \\
                      & 5 min        & 1.003 (0.008)& \textbf{1.004} (0.008)&1.003 (0.007)&1.003 (0.008)&1.001 (0.008) \\ \cmidrule(l){3-7} 
\multirow{3}{*}{SGVI} & 1 min        & \textbf{1.004} (0.001)&1.002 (0.001)&1.000 (0.001)&0.997 (0.001)&0.995 (0.001) \\
                      & 3 min        & \textbf{1.005} (0.001)&1.004 (0.001)&1.002 (0.001)&1.001 (0.001)&0.999 (0.001) \\
                      & 5 min        & \textbf{1.005} (0.001)&1.005 (0.001)&1.003 (0.001)&1.002 (0.001)&1.000 (0.001) \\ \cmidrule(l){3-7} 
\multirow{3}{*}{SGVI-c} & 1 min & \textbf{1.002} (0.001)&1.000 (0.001)&0.997 (0.001)&0.995 (0.001)&0.992 (0.001) \\
                      & 3 min        & \textbf{1.002} ($<$0.001)&1.002 (0.001)&1.000 (0.001)&0.998 (0.001)&0.996 (0.001)\\
                      & 5 min        & \textbf{1.002} ($<$0.001)&1.002 (0.001)&1.001 (0.001)&0.999 (0.001)&0.998 (0.001) \\ \cmidrule(l){3-7} 
\multirow{3}{*}{SGLD} & 1 min        & \textbf{1.001} ($<$0.001)&0.996 (0.001)&0.988 (0.002)&0.98 (0.004)&0.968 (0.005) \\
                      & 3 min        & \textbf{1.001} ($<$0.001)&0.998 (0.001)&0.991 (0.001)&0.986 (0.003)&0.977 (0.004) \\
                      & 5 min        & \textbf{1.001} ($<$0.001)&0.999 (0.001)&0.992 (0.001)&0.987 (0.002)&0.980 (0.003) \\ \bottomrule
\end{tabular}
\caption{RBODLs for SGEM, SGVI and SGLD under running times of 1, 3 and 5 minutes for the first data generation mechanism ($K=3$), with $\kappa = 0.01$ being the reference subsampling ratio. Average RBODL across 50 datasets is shown, with standard deviations in the brackets.}
\label{tab:RBODL}
\end{table*}

\paragraph{Optimal subsampling ratios.}  Table \ref{tab:RBODL} shows the RBODLs for all methods subject to three time budgets:  1, 3 and 5 minutes. We choose $\kappa = 0.01$ as the reference subsampling ratio. The results indicate that, except for SGLD run for 5 minutes, all methods reach the highest RBODL at $\kappa = 0.05$. Given that this optimun is greater than 1, this indicates that choosing a subsampling ratio around $0.05$ (rather than the baseline, $0.01$) leads to optimal model-fitting performance under a time budget. For a given method under fixed running time, we observe that the RBODL tends to drop as $\kappa$ increases. This is likely because larger subsamples take considerably more time to process due to the quadratic computational complexity for each iteration. We also observe such drops in RBODL tend to reduce in size as running times increase, which suggests better model convergence with more computation time. Finally, we see more dramatic drops in RBODL for SGVI compared to SGEM under the same running time, which suggests that the EM algorithms tends to converge faster than VI algorithms.  This result concurs with those of \cite{zhou2020} in the non-stochastic gradient setting.

\begin{table}[ht]
\centering
\begin{tabular}{@{}ccccccccccc@{}}
\toprule
methods &
 RMISE ($\boldsymbol{\alpha}, \boldsymbol{\beta}$) &
  MAE ($\boldsymbol{\mu}$) &
  IS &
  ACR &
  AIW \\ \midrule
MCMC &
  \begin{tabular}[c]{@{}c@{}}0.042\\ (0.008)\end{tabular} &
  \begin{tabular}[c]{@{}c@{}}0.072\\ (0.036)\end{tabular} &
  \begin{tabular}[c]{@{}c@{}}1.042\\ (0.274)\end{tabular} &
  \begin{tabular}[c]{@{}c@{}}0.952\\ (0.046)\end{tabular} &
  \begin{tabular}[c]{@{}c@{}}1.015\\ (0.053)\end{tabular} \\
MCMC-c &
  \begin{tabular}[c]{@{}c@{}}0.042\\ (0.008)\end{tabular} &
  \begin{tabular}[c]{@{}c@{}}0.072\\ (0.036)\end{tabular} &
  \begin{tabular}[c]{@{}c@{}}1.056\\ (0.278)\end{tabular} &
  \begin{tabular}[c]{@{}c@{}}0.952\\ (0.055)\end{tabular} &
  \begin{tabular}[c]{@{}c@{}}1.015\\ (0.058)\end{tabular} \\ \midrule
SGLD &
  \begin{tabular}[c]{@{}c@{}}0.052\\ (0.019)\end{tabular} &
  \begin{tabular}[c]{@{}c@{}}0.109\\ (0.283)\end{tabular} &
  \begin{tabular}[c]{@{}c@{}}3.898\\ (4.739)\end{tabular} &
  \begin{tabular}[c]{@{}c@{}}0.667\\ (0.152)\end{tabular} &
  \begin{tabular}[c]{@{}c@{}}0.844\\ (0.172)\end{tabular} \\
SGVI &
  \begin{tabular}[c]{@{}c@{}}0.046\\ (0.008)\end{tabular} &
  \begin{tabular}[c]{@{}c@{}}0.103\\ (0.048)\end{tabular} &
  \begin{tabular}[c]{@{}c@{}}6.163\\ (2.117)\end{tabular} &
  \begin{tabular}[c]{@{}c@{}}0.333\\ (0.131)\end{tabular} &
  \begin{tabular}[c]{@{}c@{}}0.222\\ (0.012)\end{tabular} \\
SGVI-c &
  \begin{tabular}[c]{@{}c@{}}0.040\\ (0.007)\end{tabular} &
  \begin{tabular}[c]{@{}c@{}}0.093\\ (0.044)\end{tabular} &
  \begin{tabular}[c]{@{}c@{}}4.905\\ (1.698)\end{tabular} &
  \begin{tabular}[c]{@{}c@{}}0.429\\ (0.139)\end{tabular} &
  \begin{tabular}[c]{@{}c@{}}0.213\\ (0.012)\end{tabular} \\
SGEM &
  \begin{tabular}[c]{@{}c@{}}0.100\\ (0.076)\end{tabular} &
  \begin{tabular}[c]{@{}c@{}}0.024\\ (0.026)\end{tabular} &
  - &
  - &
  -\\
SGEM-c &
  \begin{tabular}[c]{@{}c@{}}0.103\\ (0.065)\end{tabular} &
  \begin{tabular}[c]{@{}c@{}}0.023\\ (0.022)\end{tabular} &
  - &
  - &
  - \\ \midrule
{\color{black}MLE-I} &
  \begin{tabular}[c]{@{}c@{}}0.030\\ (0.007)\end{tabular} &
  \begin{tabular}[c]{@{}c@{}}0.077\\ (0.036)\end{tabular} &
  - &
  - &
  - \\
{\color{black}EM-BK}&
  \begin{tabular}[c]{@{}c@{}}0.340\\ (0.009)\end{tabular} &
  \begin{tabular}[c]{@{}c@{}}2.065\\ (0.256)\end{tabular} &
  - &
  - &
  - \\ \bottomrule
\end{tabular}\captionof{table}{Estimation metrics across all nine methods for the first data generation mechanism ($K=3$). The values in the grid cells are the average across 50 datasets, with the standard deviation in the brackets. }
\label{tab:est_compare_stoc}
\end{table}

\paragraph{Estimation accuracy.} Table \ref{tab:est_compare_stoc} shows the estimation performance measures, including RMISE, MAE, IS, ACR and AIW for all seven methods, {\color{black} along with the two frequentist benchmark methods}. Similar to the previous simulation study, we run the same algorithm on 50 datasets with 16 different initial parameter values and choose the instance associated with the highest observed data likelihood for estimation performance evaluation. We keep the same stochastic optimization hyperparameters, fixing the subsampling ratio $\kappa$ at 0.05. For SGLD, SGVI, SGVI-c, SGEM and SGEM-c, we run the algorithms for 30 minutes. For MCMC and MCMC-c, we run the algorithms on the whole dataset without subsampling for 15,000 iterations, which took around 12 hours to complete. We discard the first 5,000 samples as burn-in and calculate the posterior median for estimation performance evaluation. {\color{black} For both optimization-based frequentist methods, we run the algorithms until they have converged under the default convergence criteria.} As would be expected, the MCMC algorithms have values of RMISE, MAE and IS lower than the majority of other methods. Moreover, MCMC algorithms produce coverage rates that are very close to nominal. Among the remaining methods from Section 3, SGLD shows the best uncertainty estimation performance with the lowest IS and ACR closest to the nominal rate, while SGVI-c shows the best point estimation performance with RMISE even lower than the MCMC methods. Additionally, we observe a significant improvement in both RMISE and IS for SGVI-c compared to SGVI, indicating that incorporating a boundary correction can lead to both improved point and uncertainty estimation performance for SGVI. {\color{black}For the frequentist benchmark methods, MLE-I has the lowest RMISE among all nine methods in this simulation scenario, while having MAE higher than most methods. EM-BK has the worst point estimation accuracy across the board.} 

\paragraph{Sensitivity analysis for different dataset sizes.} We also look at how sensitive the RBODLs are to the scale of the data we have. We run the same algorithms on two sets of 50 datasets with $T=500$ and $T=2000$, and the RBODLs are shown in Tables \ref{tab:RBODL_slong} and \ref{tab:RBODL_sshort}. For the small datasets, the median RBODLs change much less than the large datasets over different subsampling ratios. This is not surprising, as it indicates that the algorithms tend to reach convergence faster for smaller datasets than the larger ones. Additionally, the optimal subsampling ratios for smaller datasets tend to be larger, indicating that there could be a fixed amount of data needed for algorithms to attain better model-fitting results.

\paragraph{Sensitivity analysis for the stochastic search parameters.} Next, we investigate the effect of the stochastic search parameters on our previous simulation results. To this end, we rerun our analysis for the medium-sized dataset under each of the following three sets of $\tau_1, \tau_2$ values: (1) $\tau_1 = 5, \tau_2 = 0.51$, (2) $\tau_1 = 1, \tau_2 = 1$, (3) $\tau_1 = 5, \tau_2 = 1$. The results are shown in Table \ref{tab:RBODL_s2}, \ref{tab:RBODL_s3} and \ref{tab:RBODL_s4}. As expected, the behavior of RBODL with respect to subsampling ratio and running time is similar to the default scenario shown in Table \ref{tab:RBODL}. Also, the results in Table \ref{tab:RBODL_s2} are more similar to those in Table \ref{tab:RBODL} than cases in Tables \ref{tab:RBODL_s3} and \ref{tab:RBODL_s4}. This is because $\tau_2$ controls the decay rate of the stepsize parameter, which has a bigger long-term effect on $\rho_r$ compared to the delay parameter $\tau_1$. We also looked at the estimation performances for all five methods under these four scenarios, with performance metrics shown in Table \ref{tab:est_compare_stoc1} and \ref{tab:est_compare2}. We can see that all algorithms performed significantly better in scenarios where $\tau_2 = 0.51$, indicating that a large value of $\tau_2$ may lead to suboptimal estimates because algorithms converged too fast. We also note that the SGEM-c outperformed SGEM where $\tau_2 = 0.51$. %in scenario Table \ref{tab:RBODL_s2}.

\paragraph{Sensitivity analysis for the threshold values of the boundary-corrected methods.} Previously, we chose a fixed value $\delta = 0.25$ as the common threshold for all boundary-corrected methods. In this simulation study, we investigate a systematic way of choosing $\delta$ based on the discussion in Section \ref{sec:approx}, and study the parameter estimation performance under different values of $\delta$. Given an estimate for $\boldsymbol{\beta}$ and a fixed value $r > 0$, we find $\delta$ such that $\delta = \frac{1}{K^2} \sum_{j=1}^K \sum_{k=1}^K \frac{1}{\hat{\beta} _{k,\ell}}$. %Intuitively, the exponential functions $\alpha _{k,\ell}\exp(-\beta _{k,\ell}(T-t))$ evaluated at $t = T - \delta$ will roughly decay to $e^{-r}$ times its boundary value, evaluated at $t = T$.
Table \ref{tab:sens_thres} shows the point and uncertainty estimation results for SGVI-c and SGEM-c under values of $r \in \{0.5, 1, 2, 3, 4\}$. For both methods, all estimation metrics reached optimality between $r=1$ and $r=2$.%, indicating that doing a first-order Taylor expansion for a certain amount of observations at the tail end of the sampled sequence may lead to lower point and uncertainty estimation errors.

{
\color{black}
\begin{table}[ht]
\centering
\begin{tabular}{@{}>{\color{black}}c>{\color{black}}c>{\color{black}}c>{\color{black}}c>{\color{black}}c>{\color{black}}cccccc@{}}
\toprule
methods &
 RMISE ($\boldsymbol{\alpha}, \boldsymbol{\beta}$) &
  MAE ($\boldsymbol{\mu}$) &
  IS &
  ACR &
  AIW \\ \midrule
  MCMC-rw &
  \begin{tabular}[c]{@{}c@{}}0.044\\ (0.008)\end{tabular} &
  \begin{tabular}[c]{@{}c@{}}0.075\\ (0.037)\end{tabular} &
  \begin{tabular}[c]{@{}c@{}}1.242\\ (0.418)\end{tabular} &
  \begin{tabular}[c]{@{}c@{}}0.951\\ (0.057)\end{tabular} &
  \begin{tabular}[c]{@{}c@{}}1.037\\ (0.063)\end{tabular} \\
  SGLD-apx &
  \begin{tabular}[c]{@{}c@{}}0.050\\ (0.012)\end{tabular} &
  \begin{tabular}[c]{@{}c@{}}0.113\\ (0.047)\end{tabular} &
  \begin{tabular}[c]{@{}c@{}}2.494\\ (1.483)\end{tabular} &
  \begin{tabular}[c]{@{}c@{}}0.763\\ (0.111)\end{tabular} &
  \begin{tabular}[c]{@{}c@{}}0.876\\ (0.120)\end{tabular} \\
 \bottomrule
\end{tabular}\captionof{table}{Estimation metrics for MCMC-rw and SGLD-apx for the first data generation mechanism ($K=3$). The values in the grid cells are the average across 50 datasets, with the standard deviation in the brackets. }
\label{tab:est_compare_stoc2}
\end{table}
}

{
\color{black}
\paragraph{Sensitivity analysis for boundary approximation.} To numerically show that the boundary approximations in \eqref{eq:ass} and \eqref{eq:corr} does not lead to a significant information loss, we implemented two additional methods as a comparison: the full MCMC algorithm with random-walk updates for $\boldsymbol{\beta}$ using the full likelihood (MCMC-rw), and the Langevin dynamics algorithm with likelihood approximation (SGLD-apx).  We focus on these two algorithms here because they are ones where both the true and the various approximate likelihoods can be implemented in a straightforward fashion. Table \ref{tab:est_compare_stoc2} shows the results for these two methods, which should be compared with those for MCMC, MCMC-c and SGLD in Table \ref{tab:est_compare_stoc1}. The results suggest that the errors introduced by the approximation are negligible (at least, compared to the Monte Carlo error involved), both in terms of both point and uncertainty estimation.
}

{
\color{black}
\paragraph{Robustness to model sparsity.} Our last evaluation uses the alternative data generation mechanism described in Section \ref{sec:syn} where we have a ten-dimensional MHP with varying degrees of sparsity in $\boldsymbol{\alpha}$.  %Similar to the three-dimensional model setting, we compared the estimation performance among the five stochastic methods over 50 independent realizations. 
Table \ref{tab:est-spars} shows the average and standard deviation of the estimation metrics for all five methods and the two frequentist benchmarks. Due to presence of sparsity, the metrics for the $\alpha_{k,\ell}$s and $\beta_{k,\ell}$s are evaluated only for entries that are associated with a non-zero true value for $\alpha_{k,\ell}$. 
{\color{black} Note that, unlike the first simulation scenario, MLE-I and EM-BK perform significantly worse than the Bayesian methods in all of our simulations with $K=10$. Furthermore, MLE-I takes much longer time to converge compared to the three-dimensional scenario. For example, the average running time for the `low sparsity' scenario is 61.47 minutes, which is much longer than the 30-minute running time for all five stochastic Bayesian methods.}  Aside from this, the results are largely consistent those from the original, three-dimensional data generation mechanism.  SGVI and SGVI-c have the best point estimation performance, SGEM and SGEM-c have the worst, and SGLD is in between. Furthermore, scenarios that are less sparse tend to be associated with lower estimation errors.} 
 
{
\color{black}
% \usepackage{multirow}
% \usepackage{multirow}
\begin{table}[]
\centering
\label{tab:est-spars}
\begin{tabular}{>{\color{black}}c>{\color{black}}c>{\color{black}}c>{\color{black}}c>{\color{black}}c>{\color{black}}c>{\color{black}}c}
\hline
methods &
  sparsity &
  RMISE ($\boldsymbol{\alpha}, \boldsymbol{\beta}$) &
  MAE ($\boldsymbol{\mu}$) &
  IS &
  ACR &
  AIW \\ \hline
\multirow{3}{*}{SGLD} &
  high &
  \begin{tabular}[c]{@{}c@{}}0.048\\ (0.005)\end{tabular} &
  \begin{tabular}[c]{@{}c@{}}0.295\\ (0.035)\end{tabular} &
  \begin{tabular}[c]{@{}c@{}}1.562\\ (0.546)\end{tabular} &
  \begin{tabular}[c]{@{}c@{}}0.742\\ (0.033)\end{tabular} &
  \begin{tabular}[c]{@{}c@{}}1.003\\ (0.085)\end{tabular} \\
 &
  medium &
  \begin{tabular}[c]{@{}c@{}}0.041\\ (0.005)\end{tabular} &
  \begin{tabular}[c]{@{}c@{}}0.256\\ (0.037)\end{tabular} &
  \begin{tabular}[c]{@{}c@{}}2.172\\ (0.680)\end{tabular} &
  \begin{tabular}[c]{@{}c@{}}0.802\\ (0.036)\end{tabular} &
  \begin{tabular}[c]{@{}c@{}}1.327\\ (0.144)\end{tabular} \\
 &
  low &
  \begin{tabular}[c]{@{}c@{}}0.031\\ (0.003)\end{tabular} &
  \begin{tabular}[c]{@{}c@{}}0.171\\ (0.018)\end{tabular} &
  \begin{tabular}[c]{@{}c@{}}2.915\\ (1.033)\end{tabular} &
  \begin{tabular}[c]{@{}c@{}}0.841\\ (0.034)\end{tabular} &
  \begin{tabular}[c]{@{}c@{}}1.443\\ (0.072)\end{tabular} \\
\multirow{3}{*}{SGVI} &
  high &
  \begin{tabular}[c]{@{}c@{}}0.047\\ (0.013)\end{tabular} &
  \begin{tabular}[c]{@{}c@{}}0.179\\ (0.038)\end{tabular} &
  \begin{tabular}[c]{@{}c@{}}2.011\\ (0.735)\end{tabular} &
  \begin{tabular}[c]{@{}c@{}}0.699\\ (0.059)\end{tabular} &
  \begin{tabular}[c]{@{}c@{}}0.680\\ (0.112)\end{tabular} \\
 &
  medium &
  \begin{tabular}[c]{@{}c@{}}0.040\\ (0.008)\end{tabular} &
  \begin{tabular}[c]{@{}c@{}}0.158\\ (0.039)\end{tabular} &
  \begin{tabular}[c]{@{}c@{}}3.158\\ (1.131)\end{tabular} &
  \begin{tabular}[c]{@{}c@{}}0.704\\ (0.053)\end{tabular} &
  \begin{tabular}[c]{@{}c@{}}0.766\\ (0.105)\end{tabular} \\
 &
  low &
  \begin{tabular}[c]{@{}c@{}}0.030\\ (0.004)\end{tabular} &
  \begin{tabular}[c]{@{}c@{}}0.111\\ (0.029)\end{tabular} &
  \begin{tabular}[c]{@{}c@{}}5.264\\ (1.269)\end{tabular} &
  \begin{tabular}[c]{@{}c@{}}0.620\\ (0.045)\end{tabular} &
  \begin{tabular}[c]{@{}c@{}}0.664\\ (0.089)\end{tabular} \\
\multirow{3}{*}{SGVI-c} &
  high &
  \begin{tabular}[c]{@{}c@{}}0.047\\ (0.012)\end{tabular} &
  \begin{tabular}[c]{@{}c@{}}0.183\\ (0.038)\end{tabular} &
  \begin{tabular}[c]{@{}c@{}}1.904\\ (0.729)\end{tabular} &
  \begin{tabular}[c]{@{}c@{}}0.703\\ (0.062)\end{tabular} &
  \begin{tabular}[c]{@{}c@{}}0.665\\ (0.047)\end{tabular} \\
 &
  medium &
  \begin{tabular}[c]{@{}c@{}}0.039\\ (0.008)\end{tabular} &
  \begin{tabular}[c]{@{}c@{}}0.163\\ (0.039)\end{tabular} &
  \begin{tabular}[c]{@{}c@{}}2.941\\ (1.051)\end{tabular} &
  \begin{tabular}[c]{@{}c@{}}0.710\\ (0.055)\end{tabular} &
  \begin{tabular}[c]{@{}c@{}}0.747\\ (0.051)\end{tabular} \\
 &
  low &
  \begin{tabular}[c]{@{}c@{}}0.029\\ (0.004)\end{tabular} &
  \begin{tabular}[c]{@{}c@{}}0.118\\ (0.029)\end{tabular} &
  \begin{tabular}[c]{@{}c@{}}4.515\\ (1.110)\end{tabular} &
  \begin{tabular}[c]{@{}c@{}}0.637\\ (0.048)\end{tabular} &
  \begin{tabular}[c]{@{}c@{}}0.638\\ (0.033)\end{tabular} \\
\multirow{3}{*}{SGEM} &
  high &
  \begin{tabular}[c]{@{}c@{}}0.047\\ (0.007)\end{tabular} &
  \begin{tabular}[c]{@{}c@{}}0.308\\ (0.033)\end{tabular} &
  - &
  - &
  - \\
 &
  medium &
  \begin{tabular}[c]{@{}c@{}}0.039\\ (0.005)\end{tabular} &
  \begin{tabular}[c]{@{}c@{}}0.263\\ (0.030)\end{tabular} &
  - &
  - &
  - \\
 &
  low &
  \begin{tabular}[c]{@{}c@{}}0.030\\ (0.003)\end{tabular} &
  \begin{tabular}[c]{@{}c@{}}0.160\\ (0.026)\end{tabular} &
  - &
  - &
  - \\
\multirow{3}{*}{SGEM-c} &
  high &
  \begin{tabular}[c]{@{}c@{}}0.047\\ (0.006)\end{tabular} &
  \begin{tabular}[c]{@{}c@{}}0.263\\ (0.034)\end{tabular} &
  - &
  - &
  - \\
 &
  medium &
  \begin{tabular}[c]{@{}c@{}}0.039\\ (0.005)\end{tabular} &
  \begin{tabular}[c]{@{}c@{}}0.161\\ (0.030)\end{tabular} &
  - &
  - &
  - \\
 &
  low &
  \begin{tabular}[c]{@{}c@{}}0.030\\ (0.003)\end{tabular} &
  \begin{tabular}[c]{@{}c@{}}0.161\\ (0.027)\end{tabular} &
  - &
  - &
  - \\ \midrule
  \multirow{3}{*}{MLE-I} &
  high &
  \begin{tabular}[c]{@{}c@{}}0.206\\ (0.015)\end{tabular} &
  \begin{tabular}[c]{@{}c@{}}0.331\\ (0.151)\end{tabular} &
  - &
  - &
  - \\
 &
  medium &
  \begin{tabular}[c]{@{}c@{}}0.155\\ (0.008)\end{tabular} &
  \begin{tabular}[c]{@{}c@{}}0.421\\ (0.199)\end{tabular} &
  - &
  - &
  - \\
 &
  low &
  \begin{tabular}[c]{@{}c@{}}0.124\\ (0.003)\end{tabular} &
  \begin{tabular}[c]{@{}c@{}}0.247\\ (0.195)\end{tabular} &
  - &
  - &
  - \\
    \multirow{3}{*}{EM-BK} &
 
 high &
  \begin{tabular}[c]{@{}c@{}}0.462\\ (0.047)\end{tabular} &
  \begin{tabular}[c]{@{}c@{}}0.584\\ (0.106)\end{tabular} &
  - &
  - &
  - \\
 &
  medium &
  \begin{tabular}[c]{@{}c@{}}0.366\\ (0.048)\end{tabular} &
  \begin{tabular}[c]{@{}c@{}}0.445\\ (0.106)\end{tabular} &
  - &
  - &
  - \\
 &
  low &
  \begin{tabular}[c]{@{}c@{}}0.261\\ (0.019)\end{tabular} &
  \begin{tabular}[c]{@{}c@{}}0.197\\ (0.088)\end{tabular} &
  - &
  - &
  - \\
  \hline
\end{tabular}
\caption{Estimation metrics across SGLD, SGVI, SGVI-c, SGEM , SGEM-c, MLE-I and EM-BK for the second data generation mechanism ($K=10$) with three levels of model sparsity. The values in the grid cells are the average across 50 datasets, with the standard deviation in the brackets.}
\label{tab:est-spars}
\end{table}
}

\section{Real-world application}

\paragraph{Data description.} In this section, we apply our methods to model the market risk dynamics in the Standard \& Poor (S\&P)'s 500 intraday index prices for its 11 sectors: Consumer Discretionary (COND), Communication Staples (CONS), Energy (ENRS), Financials (FINL), Health Care (HLTH), Industrials (INDU), Information Technology (INFT), Materials (MATR), Real Estate (RLST), Communication Services (TELS), Utilities (UTIL). To achieve this, price data between August 22, 2022 and Jan 23, 2023 was downloaded from Bloomberg Finance L.P. Similar to \cite{rodriguez2017assessing}, an event occurs on dimension $k=1,\ldots, 11$ if the negative log returns in sector $k$ exceeds a predetermined threshold (in our case, a 0.05\% drop on a one-minute basis). The resulting dataset contains 55,509 events across the 11 dimensions. 

\paragraph{Results.} We fit a Hawkes process model with exponential decay functions to the event data using the SGEM, SGEM-c, SGVI, SGVI-c and SGLD algorithms. We set the subsampling ratio of $\kappa = 0.01$ for SGLD and of $\kappa = 0.05$ for all other methods. Similar to the procedure in Section \ref{sec:syn}, we start all algorithms at 16 different initial points and choose the instances with the highest observed data likelihood to compute the estimates.  Furthermore, all these algorithms were run for a fixed period of 30 minutes for each initial set of values. As a reference, we also apply MCMC and MCMC-c to the dataset and use 10,000 posterior samples after 10,000 burn-ins, which roughly took around two days. %Similar to the procedure in Section \ref{sec:syn}, we start all algorithms at 16 initial points and choose the instances with the highest observed data likelihood to compute the estimates. %We can observe if there are clustering structures among the index sectors, i.e. whether there are significant mutual excitation among particular sets of sectors. In particular, we develop an exploratory approach to visualize the clustering structure based on the multidimensional scaling algorithm (MDS) \cite{torgerson1952multidimensional}.

Figure \ref{fig:mat_alpha} shows heatmaps of point estimates for the $\boldsymbol{\alpha}$ parameters for all seven algorithms. To facilitate comparisons, we also generate a visual representation by constructing a measure of similarity between sectors $i$ and $j$ as $\Upsilon(i,j) = \exp\left\{-\frac{1}{2}(\alpha_{ij} + \alpha_{ji})\right\}$, and then use multidimensional scaling \citep{torgerson1952multidimensional} to find a two-dimensional representation of these similarities. Because the representation is arbitrary up to translations, rotations and reflections, we use Procrustes analysis \citep{dryden2016statistical} to align the representations for all the algorithms. %We then compute the first two principal coordinates generated from the MDS algorithm and applied the Procrustes analysis algorithm \cite{dryden2016statistical} to the matrices of principal coordinates in order to compare across different methods on the same basis. {\color{red} I am not sure if this is a good description of the procedure.} The rescaled principal coordinates for all seven methods are shown in the right panel in Figure \ref{fig:mat_alpha}.
All seven methods yield similar point estimates for $\boldsymbol{\alpha}$. To explore this question in more detail, we present in Table \ref{tab:app_dist_alpha} the mean square distance between point estimates for each pair of methods.  We see that MCMC and MCMC-c are almost identical, and that SGEM, SGVI, and SGEM-c yield very similar estimates.  Interestingly, SGVI-c yields results that are as close to those of the MCMC ``gold standard'' as those from SGEM, SGVI, and SGEM-c, but that are fairly different from them.  We also note that SGLD seems to yield results that are the furthest away from the MCMC procedures, suggesting that a time budget of 30 minutes is not enough to achieve reliable results in this example.  From a substantive point of view, Figure \ref{fig:mat_alpha} suggests mutual excitation of exceedances within each of the following three groups: (1) UTIL, MATR COND, (2) INFT, FINL and INDU, (3) TELS, RLST, HLTH and CONS and ENRS.  One particular interesting result is the estimates for the energy (ENRS) sector, which has a much higher diagonal $\alpha$ estimate and lower off diagonal estimates corresponding to other sectors. This is supported by the scatterplot of principal coordinates, in which the point for ENRS is away from all other sectors, indicating that such sector may be less likely associated with price movements in other sectors.

Next, we show in Figure \ref{fig:mat_beta} point estimates of $\boldsymbol{\beta}$ under all 7 methods and, in Table \ref{tab:app_dist_beta}, the mean square distance between the estimates generated by the different methods. The pattern of the results is very similar: (1) MCMC and MCMC-c yield the most similar results, (2) SGLD seems to yield estimates that are furthest away from those generated by the MCMC methods, (3) SGEM, SGVI, and SGEM-c yield very similar results to each other, and (4) SGVI-c yields different results from SGEM, SGVI, and SGEM-c, but they are as close to those of the MCMC approaches as those from the three alternatives.  We note, however, that the estimates of $\boldsymbol{\beta}$ generated by MCMC and MCMC-c do seem to differ from each other much more than than the estimates of $\boldsymbol{\alpha}$ did.

Figure \ref{fig:mat_mu} shows the point estimates for $\boldsymbol{\mu}$, and Table \ref{tab:app_dist_mu} shows mean square distances between the model estimates.  Not surprisingly, the same patterns arise again, although we note that the distances tend to be smaller.  From an application point of view, we note that all methods identify ENRS as a sector with a very high baseline rate events, FINL, INFT, INDU, HLTH and CONS as sectors where the majority of price drops are the result of contagion from turbulence in other sectors.

As for uncertainty estimation, Figures \ref{fig:app_unc_alpha}, \ref{fig:app_unc_beta} and \ref{fig:app_unc_mu} show the length of the estimated posterior credible intervals for $\boldsymbol{\alpha}$, $\boldsymbol{\beta}$ and $\boldsymbol{\mu}$ for SGVI, SGVI-c and SGLD, as well as for MCMC and MCMC-c.  As was the case with simulated datasets, stochastic gradient methods seem to underestimate the uncertainty in the posterior distribution, with SGVI and SGVI-c doing much more dramatically than SGLD. 

{
\color{black} 
Finally, we conduct in-sample goodness-of-fit analysis using quantile-quantile plots for the posterior distribution of inter-event times. The construction of these plots relies on the well-known time-rescaling theorem \citep{daley2008}, which states that, if the Hawkes model process is correct, the transformed inter-arrival times on dimension $\ell$, $z_{i}^\ell = 1 - \exp \left\{-\left[\Lambda_\ell \left(t_{i}^{\ell}\right)-\Lambda_\ell \left(t_{i-1}^{\ell}\right)\right]\right\}$, where $\Lambda_\ell(t) = \int_0^t \lambda_\ell (s)ds$ is the compensator for $\lambda_{\ell}(t)$), follow a uniform distribution on the unit interval. Quantile-quantile plots of the (transformed) observed inter-arrival vs.\ those for the uniform distribution for all seven methods are shown in the Supplementary Materials. All seven methods produce very similar results. In all cases, the plots suggests that a MHP model tends to predict somewhat shorter inter-arrival times than expected.}

\section{Discussion}

Our experiments indicate that computational methods for Hawkes process models based on stochastic gradients can lead to accurate estimators and substantially reduce the computational burden in large datasets. However, our results also indicate some clear tradeoffs. SGEM algorithms are the fastest (which is consistent with the results of \cite{zhou2020} for full-batch methods), but they do not yield interval estimates of the parameters. SGVI algorithms are almost as computationally efficient as SGEM and yield interval estimates. However, these interval estimates are too narrow, leading to substantial undercoverage. That variational inference underestimates the variance of the posterior distribution is well known (e.g., see \citealp{blei2017variational}), but it was still striking to see how low the coverage can be in the case of MHPs.  SGLD algorithms are the slowest and require careful tuning, but can also lead to more accurate interval estimates if allowed to run for enough time.

Our experiments also suggest that the new approximation to the full-data likelihood based on a first-order Taylor expansion of the compensator of the Hawkes process has the potential to improve the accuracy of the algorithm with minimal additional computational costs.  This was clearer for SGVI algorithms, where the approximation clearly improved the MSE of the point estimators.  Finally, our experiments suggest that, as sample sizes grow, the fraction of time involved in the subsamples used to compute stochastic gradients can decrease as long as the number of observations in each subsample remains above a critical threshold. This is important, because it suggests that, at least empirically, the computational complexity of this algorithms can remain roughly constant as a function of the sample size.

{\color{black} As pointed out by one of the referees, an often important question in the context of MHPs is to identify which interactions are non-null.  Bayesian inference for the the interaction graph can be accommodated through a slight modification of the priors in Section \ref{sec:hyperpar} in which the single Gamma prior for $\alpha_{k,\ell}$ is replaced by a ``spike-and-slab’’ style specification (e.g., see \citealp{mitchell1988bayesian}) which, in this case, can be defined through a mixture of two Gamma priors with very different means.  Extending our algorithms to this setting is straightforward, perhaps requiring the introduction of additional latent component indicators for each of the $\alpha_{k,\ell}$s.  Our numerical explorations of this variant of the model suggest that the SGEM and SGVI algorithms perform quite well in this setting, but the SGLD algorithm struggles to properly explore the parameter space.
Another important question is how these methods extend to non-linear Hawkes processes (e.g., see \citealp{bremaud1996stability}), which allows for the $\alpha_{k,\ell}$s to take negative values. With the addition of a nonlinear activation layer in the intensity function, nonlinear MHP models break model conjugacy and therefore cause computational challenges for SGEM and SGVI. For certain non-linear MHP models, latent variable augmentation can be used to restore model conjugacy, enabling a more or less direct extension of the algorithms described here (e.g., see \citealp{malem2021nonlinear, zhou2021nonlinear}).  Alternatives that could be used in more general settings include methods based on Laplace approximations (e.g., see \citep{wang2013variational}) and non-conjuage message passing (e.g., see \citealp{khan2017conjugate}), which  have been successfully applied  in other model settings. However, as far we are aware, these have not yet been explored for general non-linear Hawkes process.  On the other hand, since SGLD does not rely on conjugacy, its extension to nonlinear MHP is relatively more straightforward than SGEM and SGVI.   
}

The work in this manuscript focused on a very particular class of MHP with constant baseline intensity and a parametric excitation function. This was a deliberate choice meant to simplify exposition and interpretation.  However, the insights from this manuscript apply much more broadly.  For example, we are currently work on fast inference algorithms for MHP models where the excitation functions are modeled non parametrically using mixtures of dependent Dirichlet processes.  This, and other extensions, will be discussed elsewhere.

\section*{Reproducibility information}

The R scripts for the algorithms were run on CentOS 7 Linux, with 128GB of memory. The dataset and the R code for the simulation and application examples can be found at \url{https://github.com/AlexJiang1125/MHP}.

\section*{Acknowledgements}

This research was partially supported by NSF grants NSF-2023495 and NSF-2114727. We would also like to acknowledge the support of the Research Computing Club at the University of Washington by providing access to their computational resources.  We also express our gratitude to two referees, whose feedback during the review process helped improve the manuscript. 

\appendix  
{
\color{black}
\section{Loss of information from the boundary-corrected likelihood approximation}
\label{ap:theorem}

Given the time domain $(0,T]$ and approximation threshold $\delta$, We define $\varphi_{k,\ell}(\delta, T)$ as the \textbf{information loss ratio} for approximating $\Lambda_{k, \ell} (\cdot)$, which has the following formula: 
\begin{align}
    \varphi_{k, \ell}(\delta, T) = \frac{1}{T} \int_0^T \left[1 - \frac{1 - \exp \left(- \beta_{k, \ell}(T-t) \right)}{\mathbf{1}(0 < t \leq T-\delta) + \left[\beta_{k,\ell}(T-t) \right]\cdot \mathbf{1}(T-\delta < t \leq T)} \right] dt.
\end{align}

We can interpret $\varphi_{k,\ell}(\delta, T)$ as follows: suppose $t$ is a time event on dimension $k$, $\varphi_{k,\ell}(\delta, T)$ is one minus the ratio of $t$'s contribution to the compensator (sans the common term $\mu_k T$), integrated over a uniform distribution for $t$ on $[0,T]$.  The use of a uniform distribution here seems appropriate under the assumption that the MHP is stationary.

\newtheorem{theorem}{Theorem}

\begin{theorem}
Assuming $\max_{k, \ell} |\beta_{k,\ell}| \leq M$ and $\delta$ is fixed, the information loss ratio converges to zero as $T \rightarrow \infty$, i.e.
\begin{align*}
    \lim_{T \rightarrow \infty} \max_{1 \leq k, \ell \leq K} |\varphi_{k, \ell}(\delta, T)| = 0.
\end{align*}
\end{theorem}

\begin{proof}
We proved the theorem by showing that $\lim_{T \rightarrow \infty} |\varphi_{k, \ell}(\delta, T)| = 0$ for all $k, \ell = 1, \dots, K$. To derive an upper bound on $|\varphi_{k, \ell}(\delta, T)|$, we start by breaking down the integration domain into two parts: $(0, T-\delta]$ and $(T-\delta, T]$. For the first part, we have:

\begin{align*}
    \left|\frac{1}{T} \int_0^{T-\delta} \left[1 - (1 - \exp (-\beta_{k, \ell} (T-t))) \right] dt\right| = \left|\frac{1}{T} \int_{0}^{T-\delta} \exp (-\beta_{k, \ell}t) dt \right| &=  \frac{1}{\beta_{k, \ell} T}\left[1 - \exp (-\beta_{k,\ell}(T-\delta)) \right].
\end{align*}
For the second part, we have: 
\begin{align*}
    \left|\frac{1}{T} \int_{T-\delta}^T \left[ 1 - \frac{1 - \exp \left(- \beta_{k, \ell}(T-t) \right)}{\left[\beta_{k,\ell}(T-t) \right]} \right] dt \right| &= \frac{1}{\beta_{k, \ell} T} \int_0^{\beta_{k, \ell} \delta} \frac{s-1-\exp(-s)}{s}ds \\
    &= \frac{1}{\beta_{k, \ell} T} \int_0^{\beta_{k, \ell} \delta} \frac{\exp(- 2 \Delta s) s^2}{2s}ds, ~~ (0 \leq \Delta  s \leq s) \\
    &\leq  \frac{1}{\beta_{k, \ell} T}\int_0^{\beta_{k, \ell} \delta} \frac{s}{2}ds = \frac{\beta_{k, \ell} \delta^2}{4T}.
\end{align*}

Assuming $\delta$ is fixed and $\beta_{k, \ell}$ is bounded by a fixed constant over $k, \ell = 1, \dots , K$, both parts converge to zero as $T  \rightarrow \infty$. As $|\varphi_{k, \ell}(\delta, T)|$ is upper bounded by the two parts above, we showed that $\lim_{T \rightarrow \infty} |\varphi_{k, \ell}(\delta, T)| = 0$. 
\end{proof}
}
\begin{comment}
Outline for the discussion section: 

\begin{itemize}
    \item Among the three stochastic algorithms that we consider, SGEM is the fastest algorithm to converge, followed by SGVI and SGLD.
    \item From the simulation studies SGLD gives the best estimation results followed by SGVI and SGEM. There is a tradeoff between computation time and estimation accuracy
\end{itemize}

In the paper, we proposed a novel framework for developing scalable inference algorithm for MHP models based on the stochastic optimization technique. Additionally, we looked into a commonly used likelihood approximation technique employed by previous literature and proposed an extension. Our simulation studies show that such extension leads to improvements in both point and uncertainty estimation performances. We also demonstrate that the SGLD algorithm shows superior estimation accuracy compared to SGEM and SGVI methods, which we believe could be an interesting lane for future research.
\end{comment}

\bibliographystyle{apalike}
\bibliography{articles}

\newpage

%\appendix
%\onecolumn
%\section{Appendices}
%
%\subsection{Extra Figures and Tables}

\begin{table*}[t]
\centering
\begin{tabular}{@{}ccccccc@{}}
\toprule
methods               & running time & 0.05         & 0.1          & 0.2          & 0.3          & 0.4          \\ \midrule
\multirow{3}{*}{SGEM} & 1 min        & \textbf{0.999} (0.002)&0.998 (0.002)&0.994 (0.002)&0.993 (0.002)&0.991 (0.003) \\
                      & 3 min        & \textbf{1.001} (0.003)&0.999 (0.003)&0.997 (0.003)&0.996 (0.003)&0.994 (0.003) \\
                      & 5 min        & \textbf{1.001} (0.009)&1.000 (0.008)&0.999 (0.008)&0.997 (0.009)&0.996 (0.009) \\ \cmidrule(l){3-7} 
\multirow{3}{*}{SGEM-c} & 1 min     & \textbf{1.003} (0.010)&0.998 (0.017)&0.996 (0.022)&1.000 (0.025)&1.000 (0.027) \\
                      & 3 min        & \textbf{1.006} (0.006)&1.003 (0.01)&0.994 (0.018)&0.994 (0.021)&0.998 (0.023) \\
                      & 5 min        & \textbf{1.007} (0.007)&1.006 (0.009)&1.003 (0.015)&1.002 (0.02)&1.003 (0.023) \\ \cmidrule(l){3-7} 
\multirow{3}{*}{SGVI} & 1 min        & \textbf{0.999} ($<$ 0.001)&0.996 (0.001)&0.991 (0.001)&0.988 (0.001)&0.986 (0.002) \\
                      & 3 min        & \textbf{1.001} ($<$ 0.001)&0.999 ($<$ 0.001)&0.995 (0.001)&0.992 (0.001)&0.99 (0.001)\\
                      & 5 min        & \textbf{1.001} ($<$ 0.001)&1.000 ($<$ 0.001)&0.997 ($<$ 0.001)&0.994 (0.001)&0.992 (0.001) \\ \cmidrule(l){3-7} 
\multirow{3}{*}{SGVI-c}& 1 min        & \textbf{0.998} ($<$ 0.001)&0.995 (0.001)&0.991 (0.001)&0.988 (0.001)&0.985 (0.002) \\
                      & 3 min        & \textbf{1.000} ($<$ 0.001)&0.998 ($<$ 0.001)&0.994 (0.001)&0.992 (0.001)&0.99 (0.001) \\
                      & 5 min        & \textbf{1.000} ($<$ 0.001)&0.999 ($<$ 0.001)&0.996 (<0.001)&0.993 (0.001)&0.991 (0.001) \\ \cmidrule(l){3-7} 
\multirow{3}{*}{SGLD} & 1 min        & \textbf{0.997} (0.001)&0.991 (0.002)&0.98 (0.003)&0.96 (0.005)&0.945 (0.010) \\
                      & 3 min        & \textbf{0.999} (0.001)&0.993 (0.001)&0.985 (0.002)&0.973 (0.003)&0.959 (0.005)\\ 
                      & 3 min        & \textbf{0.999} (0.001)&0.993 (0.001)&0.985 (0.002)&0.973 (0.003)&0.959 (0.005)\\ \bottomrule
\end{tabular}
\caption{Sensitivity analysis for the dataset sizes, with a large dataset $(T = 2000)$ for the first data generation mechanism ($K=3$). RBODLs for SGEM, SGVI and SGLD under running times of 1, 3 and 5 minutes, with $\kappa = 0.01$ being the reference subsampling ratio. Average RBODL across 50 datasets is shown, with standard deviations in the brackets.}
\label{tab:RBODL_slong}
\end{table*}

\begin{table*}[t]
\centering
\begin{tabular}{@{}ccccccc@{}}
\toprule
methods               & running time & 0.05         & 0.1          & 0.2          & 0.3          & 0.4          \\ \midrule
\multirow{3}{*}{SGEM} & 1 min        & \textbf{1.008} (0.008)&1.008 (0.005)&1.008 (0.005)&1.007 (0.005)&1.006 (0.005) \\
                      & 3 min        & 1.007 (0.007)& \textbf{1.008} (0.006)&1.008 (0.007)&1.007 (0.006)&1.007 (0.007) \\
                      & 5 min        & 1.008 (0.01)&1.008 (0.009)& \textbf{1.009} (0.009)&1.008 (0.009)&1.008 (0.009) \\ \cmidrule(l){3-7} 
\multirow{3}{*}{SGEM-c} & 1 min        & \textbf{1.003} (0.012)&1.003 (0.012)&1.002 (0.011)&1.001 (0.011)&1.001 (0.011) \\
                      & 3 min        & \textbf{1.003} (0.007)&1.003 (0.007)&1.004 (0.007)&1.003 (0.007)&1.003 (0.007) \\
                      & 5 min        & 1.002 (0.006)& \textbf{1.003} (0.006)&1.003 (0.005)&1.002 (0.006)&1.002 (0.006) \\ \cmidrule(l){3-7} 
\multirow{3}{*}{SGVI} & 1 min        & \textbf{1.032} (0.009)&1.032 (0.009)&1.03 (0.009)&1.029 (0.009)&1.027 (0.009) \\
                      & 3 min        & \textbf{1.032} (0.009)&1.032 (0.009)&1.032 (0.009)&1.031 (0.009)&1.03 (0.009)\\
                      & 5 min        & \textbf{1.033} (0.009)&1.033 (0.009)&1.033 (0.009)&1.032 (0.009)&1.032 (0.009) \\ \cmidrule(l){3-7} 
\multirow{3}{*}{SGVI-c}& 1 min        & \textbf{1.021} (0.007)&1.021 (0.007)&1.019 (0.007)&1.018 (0.007)&1.016 (0.007) \\
                      & 3 min        & \textbf{1.022} (0.007)&1.022 (0.007)&1.022 (0.007)&1.021 (0.007)&1.02 (0.007) \\
                      & 5 min        & \textbf{1.022} (0.007)&1.022 (0.007)&1.022 (0.007)&1.022 (0.007)&1.021 (0.007) \\ \cmidrule(l){3-7} 
\multirow{3}{*}{SGLD} & 1 min        & \textbf{1.017} (0.006)&1.015 (0.006)&1.008 (0.005)&1.004 (0.005)&0.997 (0.004) \\
                      & 3 min        & \textbf{1.019} (0.006)&1.018 (0.006)&1.014 (0.006)&1.009 (0.005)&1.005 (0.005) \\                      
                      & 5 min        & \textbf{1.020} (0.006)&1.020 (0.006)&1.017 (0.006)&1.011 (0.005)&1.007 (0.005) \\ \bottomrule
\end{tabular}
\caption{Sensitivity analysis for the dataset sizes, with a large dataset $(T = 500)$ for the first data generation mechanism ($K=3$). RBODLs for SGEM, SGVI and SGLD under running times of 1, 3 and 5 minutes, with $\kappa = 0.01$ being the reference subsampling ratio. Average RBODL across 50 datasets is shown, with standard deviations in the brackets.}
\label{tab:RBODL_sshort}
\end{table*}

\begin{table}[!ht]
\centering
\begin{tabular}{@{}ccccccc@{}}
\toprule
methods               & running time & 0.05         & 0.1          & 0.2          & 0.3          & 0.4          \\ \midrule
\multirow{3}{*}{SGEM} & 1 min        & \textbf{1.003} (0.007)&1.002 (0.005)&1.000 (0.005)&0.997 (0.005)&0.995 (0.005) \\
                      & 3 min        & \textbf{1.004} (0.006)&1.003 (0.005)&1.002 (0.006)&1.000 (0.006)&0.998 (0.006) \\
                      & 5 min        & 1.003 (0.008)& \textbf{1.004} (0.008)&1.003 (0.008)&1.001 (0.008)&1.000 (0.008) \\ \cmidrule(l){3-7} 
\multirow{3}{*}{SGEM-c} & 1 min     & \textbf{1.003} (0.005)&1.003 (0.005)&1.000 (0.004)&0.998 (0.004)&0.996 (0.005) \\
                      & 3 min        & 1.003 (0.009)& \textbf{1.004} (0.008)&1.003 (0.008)&1.001 (0.008)&1.001 (0.008) \\
                      & 5 min        & \textbf{1.004} (0.006)&1.003 (0.006)&1.003 (0.006)&1.002 (0.006)&1.001 (0.006) \\ \cmidrule(l){3-7} 
\multirow{3}{*}{SGVI} & 1 min        & \textbf{1.003} (0.005)&1.001 (0.005)&0.998 (0.005)&0.996 (0.005)&0.993 (0.005) \\
                      & 3 min        & 1.003 (0.006) & \textbf{1.004} (0.006)&1.002 (0.006)&1.000 (0.007)&0.998 (0.006) \\
                      & 5 min        & \textbf{1.004} (0.006)&1.004 (0.006)&1.003 (0.006)&1.001 (0.006)&1.000 (0.006) \\ \cmidrule(l){3-7} 
\multirow{3}{*}{SGVI-c}& 1 min        & \textbf{1.003} (0.003)&1.003 (0.002)&1.000 (0.002)&0.998 (0.002)&0.996 (0.002) \\
                      & 3 min        & \textbf{1.004} (0.006)&1.004 (0.005)&1.003 (0.005)&1.002 (0.005)&1.000 (0.005) \\
                      & 5 min        & \textbf{1.003} (0.012)&1.003 (0.012)&1.003 (0.012)&1.002 (0.012)&1.001 (0.012) \\ \cmidrule(l){3-7} 
\multirow{3}{*}{SGLD} & 1 min        & \textbf{1.004} (0.001)&1.001 (0.001)&0.996 (0.001)&0.993 (0.001)&0.99 (0.002) \\
                      & 3 min        & \textbf{1.003} (0.001)&1.003 (0.001)&1.000 (0.001)&0.997 (0.001)&0.995 (0.001) \\
                      & 5 min        & \textbf{1.003} (0.001)&1.003 (0.001)&1.001 (0.001)&0.999 (0.001)&0.997 (0.001)\\ \bottomrule
\end{tabular}
\caption{Sensitivity analysis for the stochastic search parameters for the first data generation mechanism ($K=3$), with $\tau_1 = 5, \tau_2 = 0.51$. RBODLs for SGEM, SGVI and SGLD under running times of 1, 3 and 5 minutes, with $\kappa = 0.01$ being the reference subsampling ratio. Average RBODL across 50 datasets is shown, with standard deviations in the brackets.\label{tab:RBODL_s2}}
\end{table}

\begin{table*}[!ht]
\centering
\begin{tabular}{@{}ccccccc@{}}
\toprule
methods               & running time & 0.05         & 0.1          & 0.2          & 0.3          & 0.4          \\ \midrule
\multirow{3}{*}{SGEM} & 1 min        & \textbf{1.003} (0.007)&1.002 (0.005)&1.000 (0.005)&0.997 (0.005)&0.995 (0.005) \\
                      & 3 min        & \textbf{1.004} (0.006)&1.003 (0.005)&1.002 (0.006)&1.000 (0.006)&0.998 (0.006) \\
                      & 5 min        & 1.003 (0.008)&\textbf{1.004} (0.008)&1.003 (0.008)&1.001 (0.008)&1.000 (0.008) \\ \cmidrule(l){3-7} 
\multirow{3}{*}{SGEM-c} & 1 min        & \textbf{1.003} (0.005)&1.003 (0.005)&1.000 (0.004)&0.998 (0.004)&0.996 (0.005) \\
                      & 3 min        & 1.003 (0.009)& \textbf{1.004} (0.008)&1.003 (0.008)&1.001 (0.008)&1.001 (0.008) \\
                      & 5 min        & \textbf{1.004} (0.006)&1.003 (0.006)&1.003 (0.006)&1.002 (0.006)&1.001 (0.006) \\ \cmidrule(l){3-7} 
\multirow{3}{*}{SGVI} & 1 min        & \textbf{1.003} (0.005)&1.001 (0.005)&0.998 (0.005)&0.996 (0.005)&0.993 (0.005) \\
                      & 3 min        & 1.003 (0.006)&\textbf{1.004} (0.006)&1.002 (0.006)&1.000 (0.007)&0.998 (0.006) \\
                      & 5 min        & 1.004 (0.006)& \textbf{1.004} (0.006)&1.003 (0.006)&1.001 (0.006)&1.000 (0.006) \\ \cmidrule(l){3-7} 
\multirow{3}{*}{SGVI-c}& 1 min        & \textbf{1.003} (0.003)&1.003 (0.002)&1.000 (0.002)&0.998 (0.002)&0.996 (0.002) \\
                      & 3 min        & \textbf{1.004} (0.006)&1.004 (0.005)&1.003 (0.005)&1.002 (0.005)&1.000 (0.005) \\
                      & 5 min        & \textbf{1.003} (0.012)&1.003 (0.012)&1.003 (0.012)&1.002 (0.012)&1.001 (0.012) \\ \cmidrule(l){3-7} 
\multirow{3}{*}{SGLD} & 1 min        & \textbf{1.004} (0.001)&1.001 (0.001)&0.996 (0.001)&0.993 (0.001)&0.990 (0.002) \\
                      & 3 min        & \textbf{1.003} (0.001)&1.003 (0.001)&1.000 (0.001)&0.997 (0.001)&0.995 (0.001) \\
                      & 5 min        & \textbf{1.003} (0.001)&1.003 (0.001)&1.001 (0.001)&0.999 (0.001)&0.997 (0.001)\\ \bottomrule
\end{tabular}
\caption{Sensitivity analysis for the stochastic search parameters for the first data generation mechanism ($K=3$), with $\tau_1 = 1, \tau_2 = 1$. RBODLs for SGEM, SGVI and SGLD under running times of 1, 3 and 5 minutes, with $\kappa = 0.01$ being the reference subsampling ratio. Average RBODL across 50 datasets is shown, with standard deviations in the brackets.\label{tab:RBODL_s3}}
\end{table*}

\begin{table*}[!ht]
\centering
\begin{tabular}{@{}ccccccc@{}}
\toprule
methods               & running time & 0.05         & 0.1          & 0.2          & 0.3          & 0.4          \\ \midrule
\multirow{3}{*}{SGEM} & 1 min        & \textbf{1.000} (0.006)&0.999 (0.006)&0.997 (0.007)&0.997 (0.006)&0.995 (0.006) \\
                      & 3 min        & 0.999 (0.009)&\textbf{1.000} (0.009)&0.997 (0.009)&0.997 (0.009)&0.997 (0.009) \\
                      & 5 min        & \textbf{1.000} (0.007)&0.998 (0.006)&0.998 (0.007)&0.997 (0.007)&0.996 (0.007) \\ \cmidrule(l){3-7} 
\multirow{3}{*}{SGEM-c} & 1 min     & \textbf{1.000} (0.009)&1.000 (0.010)&1.000 (0.01)&0.999 (0.011)&0.999 (0.011) \\
                      & 3 min        & \textbf{1.000} (0.005)&1.000 (0.004)&0.999 (0.005)&0.999 (0.005)&0.998 (0.006) \\
                      & 5 min        & \textbf{1.001} (0.009)&1.000 (0.009)&1.000 (0.009)&0.999 (0.009)&0.998 (0.009) \\ \cmidrule(l){3-7} 
\multirow{3}{*}{SGVI} & 1 min        & \textbf{0.998} (0.002)&0.995 (0.003)&0.991 (0.004)&0.989 (0.005)&0.987 (0.006) \\
                      & 3 min        & \textbf{0.998} (0.002)&0.996 (0.003)&0.993 (0.004)&0.992 (0.004)&0.990 (0.004) \\
                      & 5 min        & \textbf{0.999} (0.002)&0.997 (0.003)&0.995 (0.003)&0.993 (0.004)&0.991 (0.004) \\ \cmidrule(l){3-7} 
\multirow{3}{*}{SGVI-c}& 1 min      & \textbf{0.996} (0.002)&0.993 (0.003)&0.990 (0.004)&0.988 (0.005)&0.985 (0.006) \\
                      & 3 min        & \textbf{0.997} (0.002)&0.994 (0.003)&0.992 (0.004)&0.990 (0.004)&0.989 (0.005) \\
                      & 5 min        & \textbf{0.997} (0.002)&0.995 (0.003)&0.992 (0.003)&0.991 (0.004)&0.990 (0.004) \\ \cmidrule(l){3-7} 
\multirow{3}{*}{SGLD} & 1 min        & \textbf{1.001} (0.001)&1.000 (0.004)&0.997 (0.01)&0.989 (0.017)&0.972 (0.028) \\
                      & 3 min        & \textbf{1.001} (0.001)&1.000 (0.003)&0.998 (0.008)&0.995 (0.013)&0.985 (0.02)\\
                      & 5 min        & \textbf{1.001} (0.001)&1.000 (0.003)&0.998 (0.006)&0.995 (0.011)&0.990 (0.017) \\ \bottomrule
\end{tabular}
\caption{Sensitivity analysis for the stochastic search parameters for the first data generation mechanism ($K=3$), with $\tau_1 = 1, \tau_2 = 5$. RBODLs for SGEM, SGVI and SGLD under running times of 1, 3 and 5 minutes, with $\kappa = 0.01$ being the reference subsampling ratio. Average RBODL across 50 datasets is shown, with standard deviations in the brackets.\label{tab:RBODL_s4}}
\end{table*}

\begin{table}[!ht]
\centering
\begin{tabular}{@{}ccccccccccc@{}}
\toprule
 &
  \multicolumn{5}{c}{($\tau_1 = 1, \tau_2 = 0.51$)} &
  \multicolumn{5}{c}{($\tau_1 = 5, \tau_2 = 0.51$)} \\ \midrule
methods &
  \makecell{RMISE \\($\boldsymbol{\alpha}, \boldsymbol{\beta}$)} &
  MAE ($\boldsymbol{\mu}$) &
  IS &
  ACR &
  AIW &
  \makecell{RMISE \\($\boldsymbol{\alpha}, \boldsymbol{\beta}$)} &
  MAE ($\boldsymbol{\mu}$) &
  IS &
  ACR &
  AIW \\
SGLD &
  \begin{tabular}[c]{@{}c@{}}0.052\\ (0.019)\end{tabular} &
  \begin{tabular}[c]{@{}c@{}}0.109\\ (0.283)\end{tabular} &
  \begin{tabular}[c]{@{}c@{}}3.898\\ (4.739)\end{tabular} &
  \begin{tabular}[c]{@{}c@{}}0.667\\ (0.152)\end{tabular} &
  \begin{tabular}[c]{@{}c@{}}0.844\\ (0.172)\end{tabular} &
  \begin{tabular}[c]{@{}c@{}}0.053\\ (0.014)\end{tabular} &
  \begin{tabular}[c]{@{}c@{}}0.108\\ (0.046)\end{tabular} &
  \begin{tabular}[c]{@{}c@{}}3.912\\ (3.118)\end{tabular} &
  \begin{tabular}[c]{@{}c@{}}0.667\\ (0.129)\end{tabular} &
  \begin{tabular}[c]{@{}c@{}}0.907\\ (0.193)\end{tabular} \\
SGVI &
  \begin{tabular}[c]{@{}c@{}}0.046\\ (0.008)\end{tabular} &
  \begin{tabular}[c]{@{}c@{}}0.103\\ (0.048)\end{tabular} &
  \begin{tabular}[c]{@{}c@{}}6.163\\ (2.117)\end{tabular} &
  \begin{tabular}[c]{@{}c@{}}0.333\\ (0.131)\end{tabular} &
  \begin{tabular}[c]{@{}c@{}}0.222\\ (0.012)\end{tabular} &
  \begin{tabular}[c]{@{}c@{}}0.046\\ (0.008)\end{tabular} &
  \begin{tabular}[c]{@{}c@{}}0.105\\ (0.048)\end{tabular} &
  \begin{tabular}[c]{@{}c@{}}6.359\\ (2.104)\end{tabular} &
  \begin{tabular}[c]{@{}c@{}}0.333\\ (0.138)\end{tabular} &
  \begin{tabular}[c]{@{}c@{}}0.222\\ (0.012)\end{tabular} \\
SGVI-c &
  \begin{tabular}[c]{@{}c@{}}0.040\\ (0.007)\end{tabular} &
  \begin{tabular}[c]{@{}c@{}}0.093\\ (0.044)\end{tabular} &
  \begin{tabular}[c]{@{}c@{}}4.905\\ (1.698)\end{tabular} &
  \begin{tabular}[c]{@{}c@{}}0.429\\ (0.139)\end{tabular} &
  \begin{tabular}[c]{@{}c@{}}0.213\\ (0.012)\end{tabular} &
  \begin{tabular}[c]{@{}c@{}}0.040\\ (0.007)\end{tabular} &
  \begin{tabular}[c]{@{}c@{}}0.095\\ (0.044)\end{tabular} &
  \begin{tabular}[c]{@{}c@{}}4.797\\ (1.695)\end{tabular} &
  \begin{tabular}[c]{@{}c@{}}0.429\\ (0.139)\end{tabular} &
  \begin{tabular}[c]{@{}c@{}}0.216\\ (0.012)\end{tabular} \\
SGEM &
  \begin{tabular}[c]{@{}c@{}}0.049\\ (0.008)\end{tabular} &
  \begin{tabular}[c]{@{}c@{}}0.044\\ (0.042)\end{tabular} &
  - &
  - &
  - &
  \begin{tabular}[c]{@{}c@{}}0.050\\ (0.008)\end{tabular} &
  \begin{tabular}[c]{@{}c@{}}0.046\\ (0.040)\end{tabular} &
  - &
  - &
  - \\
SGEM-c &
  \begin{tabular}[c]{@{}c@{}}0.046\\ (0.008)\end{tabular} &
  \begin{tabular}[c]{@{}c@{}}0.038\\ (0.037)\end{tabular} &
  - &
  - &
  - &
  \begin{tabular}[c]{@{}c@{}}0.049\\ (0.008)\end{tabular} &
  \begin{tabular}[c]{@{}c@{}}0.044\\ (0.040)\end{tabular} &
  - &
  - &
  - \\ \bottomrule
\end{tabular}\captionof{table}{Estimation metrics across all seven methods under different sets of stochastic search parameters for the first data generation mechanism ($K=3$). The values in the grid cells are the average across 50 datasets, with the standard deviation in the brackets.
\label{tab:est_compare_stoc1}}
\end{table}

% \usepackage{booktabs}
\begin{table}[t]
\centering
%\caption{}
\begin{tabular}{@{}ccccccccccc@{}}
\toprule
 &
  \multicolumn{5}{c}{($\tau_1 = 1, \tau_2 = 1$)} &
  \multicolumn{5}{c}{($\tau_1 = 5, \tau_2 = 1$)} \\ \midrule
methods &
  \makecell{RMISE \\($\boldsymbol{\alpha}, \boldsymbol{\beta}$)} &
  RMSE($\boldsymbol{\mu}$) &
  IS &
  ACR &
  AIW &
  \makecell{RMISE \\($\boldsymbol{\alpha}, \boldsymbol{\beta}$)} &
  RMSE($\boldsymbol{\mu}$) &
  IS &
  ACR &
  AIW \\
SGLD &
  \begin{tabular}[c]{@{}c@{}}0.137\\ (0.033)\end{tabular} &
  \begin{tabular}[c]{@{}c@{}}0.264\\ (0.528)\end{tabular} &
  \begin{tabular}[c]{@{}c@{}}33.020\\ (70.245)\end{tabular} &
  \begin{tabular}[c]{@{}c@{}}0.095\\ (0.055)\end{tabular} &
  \begin{tabular}[c]{@{}c@{}}0.231\\ (0.218)\end{tabular} &
  \begin{tabular}[c]{@{}c@{}}0.150\\ (0.033)\end{tabular} &
  \begin{tabular}[c]{@{}c@{}}0.332\\ (0.121)\end{tabular} &
  \begin{tabular}[c]{@{}c@{}}37.133\\ (66.530)\end{tabular} &
  \begin{tabular}[c]{@{}c@{}}0.048\\ (0.056)\end{tabular} &
  \begin{tabular}[c]{@{}c@{}}0.250\\ (0.230)\end{tabular} \\
SGVI &
  \begin{tabular}[c]{@{}c@{}}0.158\\ (0.024)\end{tabular} &
  \begin{tabular}[c]{@{}c@{}}0.182\\ (0.061)\end{tabular} &
  \begin{tabular}[c]{@{}c@{}}41.842\\ (12.961)\end{tabular} &
  \begin{tabular}[c]{@{}c@{}}0.095\\ (0.046)\end{tabular} &
  \begin{tabular}[c]{@{}c@{}}0.291\\ (0.087)\end{tabular} &
  \begin{tabular}[c]{@{}c@{}}0.158\\ (0.025)\end{tabular} &
  \begin{tabular}[c]{@{}c@{}}0.182\\ (0.061)\end{tabular} &
  \begin{tabular}[c]{@{}c@{}}41.879\\ (13.020)\end{tabular} &
  \begin{tabular}[c]{@{}c@{}}0.095\\ (0.046)\end{tabular} &
  \begin{tabular}[c]{@{}c@{}}0.291\\ (0.088)\end{tabular} \\
SGVI-c &
  \begin{tabular}[c]{@{}c@{}}0.152\\ (0.026)\end{tabular} &
  \begin{tabular}[c]{@{}c@{}}0.176\\ (0.061)\end{tabular} &
  \begin{tabular}[c]{@{}c@{}}36.011\\ (11.889)\end{tabular} &
  \begin{tabular}[c]{@{}c@{}}0.095\\ (0.055)\end{tabular} &
  \begin{tabular}[c]{@{}c@{}}0.259\\ (0.070)\end{tabular} &
  \begin{tabular}[c]{@{}c@{}}0.152\\ (0.026)\end{tabular} &
  \begin{tabular}[c]{@{}c@{}}0.176\\ (0.061)\end{tabular} &
  \begin{tabular}[c]{@{}c@{}}36.011\\ (11.889)\end{tabular} &
  \begin{tabular}[c]{@{}c@{}}0.095\\ (0.055)\end{tabular} &
  \begin{tabular}[c]{@{}c@{}}0.259\\ (0.070)\end{tabular} \\
SGEM &
  \begin{tabular}[c]{@{}c@{}}0.087\\ (0.061)\end{tabular} &
  \begin{tabular}[c]{@{}c@{}}0.136\\ (0.062)\end{tabular} &
  - &
  - &
  - &
  \begin{tabular}[c]{@{}c@{}}0.113\\ (0.061)\end{tabular} &
  \begin{tabular}[c]{@{}c@{}}0.140\\ (0.150)\end{tabular} &
  - &
  - &
  - \\
SGEM-c &
  \begin{tabular}[c]{@{}c@{}}0.113\\ (0.061)\end{tabular} &
  \begin{tabular}[c]{@{}c@{}}0.147\\ (0.054)\end{tabular} &
  - &
  - &
  - &
  \begin{tabular}[c]{@{}c@{}}0.151\\ (0.043)\end{tabular} &
  \begin{tabular}[c]{@{}c@{}}0.114\\ (0.062)\end{tabular} &
  - &
  - &
  - \\ \bottomrule
\end{tabular}
\caption{Estimation metrics across all seven methods under different sets of stochastic search parameters for the first data generation mechanism ($K=3$). The values in the grid cells are the average across 50 datasets, with the standard deviation in the brackets.\label{tab:est_compare2}}
\end{table}

\begin{table}[ht]
\centering
%\caption{}
\begin{tabular}{@{}lllllll@{}}
\toprule
                        &         & RMISE ($\boldsymbol{\alpha}, \boldsymbol{\beta}$) & MAE ($\boldsymbol{\mu}$) & IS            & ACR           & AIW           \\ \midrule
\multirow{5}{*}{SGVI-c} & $r=0.5$ & 0.050 (0.008)                                      & 0.097 (0.048)            & 6.030 (1.420) & 0.333 (0.078) & 0.220 (0.009) \\
                        & $r=1$   & 0.040 (0.007) & 0.093 (0.044) & 4.905 (1.698)  & 0.429 (0.139) & 0.213 (0.012) \\
                        & $r=2$   & 0.039 (0.007) & 0.064 (0.038) & 2.609 (1.217)  & 0.476 (0.097) & 0.200 (0.011) \\
                        & $r=3$   & 0.052 (0.007) & 0.077 (0.037) & 5.674 (0.770)  & 0.357 (0.108) & 0.180 (0.012) \\
                        & $r=4$   & 0.080 (0.006) & 0.100 (0.045) & 13.374 (1.996) & 0.238 (0.123) & 0.161 (0.012) \\
\multirow{5}{*}{SGEM-c} & $r=0.5$ & 0.050 (0.008) & 0.038 (0.040) & -              & -             & -             \\
                        & $r=1$   & 0.046 (0.008) & 0.038 (0.037) & -              & -             & -             \\
                        & $r=2$   & 0.039 (0.007) & 0.024 (0.029) & -              & -             & -             \\
                        & $r=3$   & 0.048 (0.007) & 0.019 (0.025) & -              & -             & -             \\
                        & $r=4$   & 0.073 (0.007) & 0.028 (0.034) & -              & -             & -             \\ \bottomrule 
\end{tabular}
\caption{Estimation metrics for SGVI-c and SGEM-c under different values of $r$ for the first data generation mechanism ($K=3$). The values in the grid cells are the average across 50 datasets, with the standard deviation in the brackets.\label{tab:sens_thres}}
\end{table}

\begin{table}[ht]
\centering
\begin{tabular}{rrrrrrrr}
  \hline
 & MCMC-c & MCMC & SGVI-c & SGVI & SGEM-c & SEM & SGLD \\ 
  \hline
MCMC-c & 0.000 & 0.018 & 0.302 & 0.272 & 0.263 & 0.252 & 0.588 \\ 
  MCMC & 0.018 & 0.000 & 0.294 & 0.246 & 0.237 & 0.228 & 0.565 \\ 
  SGVI-c & 0.302 & 0.294 & 0.000 & 0.294 & 0.280 & 0.272 & 0.485 \\ 
  SGVI & 0.272 & 0.246 & 0.294 & 0.000 & 0.017 & 0.013 & 0.565 \\ 
  SGEM-c & 0.263 & 0.237 & 0.280 & 0.017 & 0.000 & 0.003 & 0.582 \\ 
  SEM & 0.252 & 0.228 & 0.272 & 0.013 & 0.003 & 0.000 & 0.573 \\ 
  SGLD & 0.588 & 0.565 & 0.485 & 0.565 & 0.582 & 0.573 & 0.000 \\ 
   \hline
\end{tabular}
\caption{Mean square distance between the point estimates for $\boldsymbol{\alpha}$ among the seven methods, computed after log transformation.}
\label{tab:app_dist_alpha}
\begin{tabular}{rrrrrrrr}
  \hline
 & MCMC-c & MCMC & SGVI-c & SGVI & SGEM-c & SEM & SGLD \\ 
  \hline
MCMC-c & 0.000 & 0.486 & 2.385 & 2.568 & 2.540 & 2.596 & 3.541 \\ 
  MCMC & 0.486 & 0.000 & 2.596 & 2.531 & 2.577 & 2.577 & 3.745 \\ 
  SGVI-c & 2.385 & 2.596 & 0.000 & 2.373 & 2.264 & 2.375 & 4.290 \\ 
  SGVI & 2.568 & 2.531 & 2.373 & 0.000 & 0.078 & 0.053 & 3.435 \\ 
  SGEM-c & 2.540 & 2.577 & 2.264 & 0.078 & 0.000 & 0.021 & 3.267 \\ 
  SEM & 2.596 & 2.577 & 2.375 & 0.053 & 0.021 & 0.000 & 3.462 \\ 
  SGLD & 3.541 & 3.745 & 4.290 & 3.435 & 3.267 & 3.462 & 0.000 \\ 
   \hline
\end{tabular}
\caption{Mean square distance between the point estimates for $\boldsymbol{\beta}$ among the seven methods, computed after log transformation.}
\label{tab:app_dist_beta}
\begin{tabular}{rrrrrrrr}
  \hline
 & MCMC-c & MCMC & SGVI-c & SGVI & SGEM-c & SEM & SGLD \\ 
  \hline
MCMC-c & 0.000 & 0.007 & 0.090 & 0.097 & 0.142 & 0.129 & 0.219 \\ 
  MCMC & 0.007 & 0.000 & 0.114 & 0.115 & 0.175 & 0.161 & 0.212 \\ 
  SGVI-c & 0.090 & 0.114 & 0.000 & 0.061 & 0.083 & 0.072 & 0.270 \\ 
  SGVI & 0.097 & 0.115 & 0.061 & 0.000 & 0.045 & 0.034 & 0.324 \\ 
  SGEM-c & 0.142 & 0.175 & 0.083 & 0.045 & 0.000 & 0.001 & 0.466 \\ 
  SEM & 0.129 & 0.161 & 0.072 & 0.034 & 0.001 & 0.000 & 0.435 \\ 
  SGLD & 0.219 & 0.212 & 0.270 & 0.324 & 0.466 & 0.435 & 0.000 \\ 
   \hline
\end{tabular}
\caption{Mean square distance between the point estimates for $\boldsymbol{\mu}$ among the seven methods, computed after log transformation.}
\label{tab:app_dist_mu}
\end{table}

\begin{figure*}[t]
	\centering
	\includegraphics[width = \linewidth]{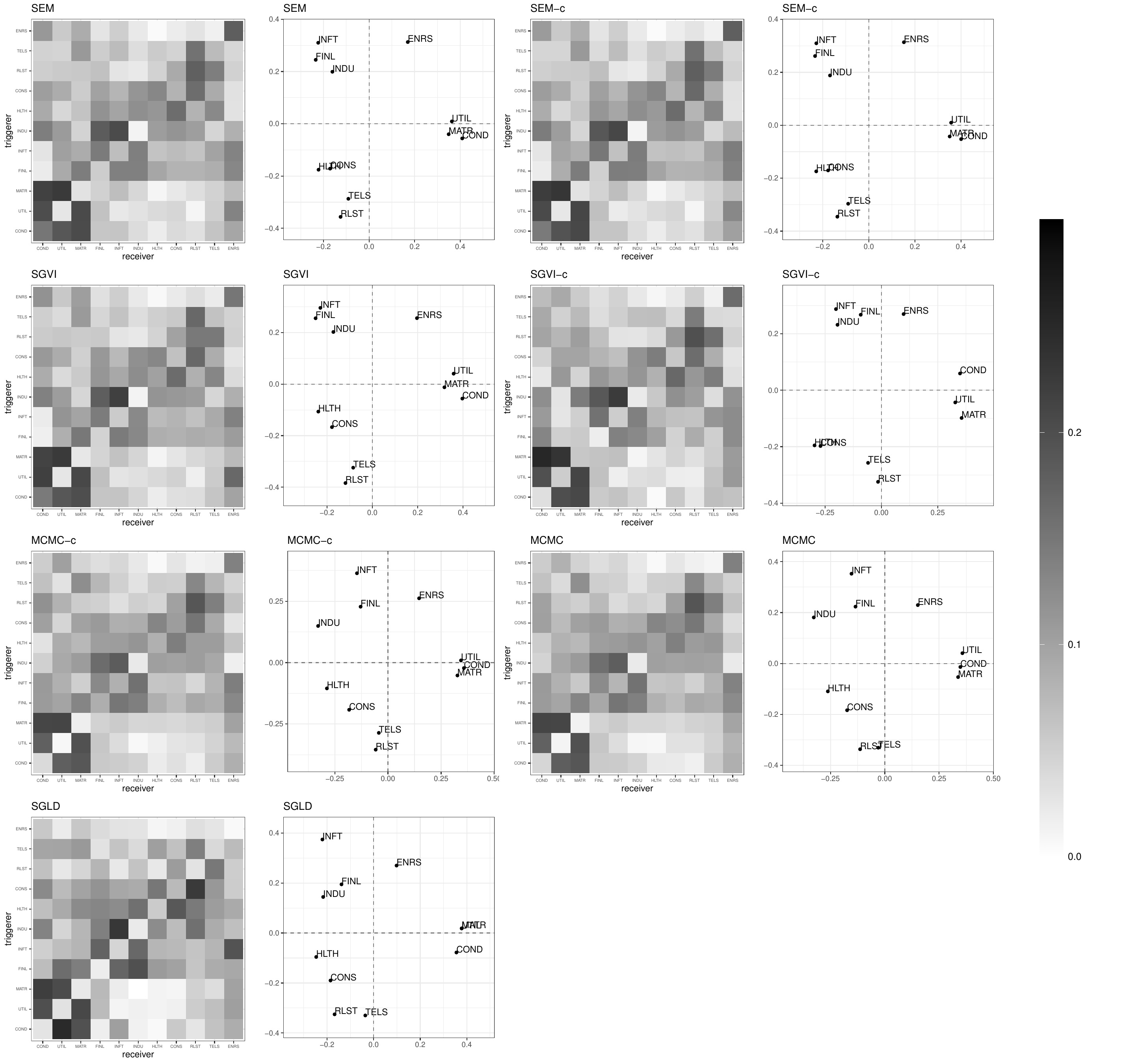}
	\caption{Left panel: heatmaps for point estimates of $\boldsymbol{\alpha}$ parameters for the corresponding 11 sectors, for all seven algorithms. Right panel: first two principal coordinates (after applying the Procrustes analysis algorithm) for the 11 sectors based on the distance measure matrix for $\boldsymbol{\alpha}$ estimates.\label{fig:mat_alpha}}
\end{figure*}

\begin{figure*}[t]
	\centering
	\includegraphics[width = 0.65\linewidth]{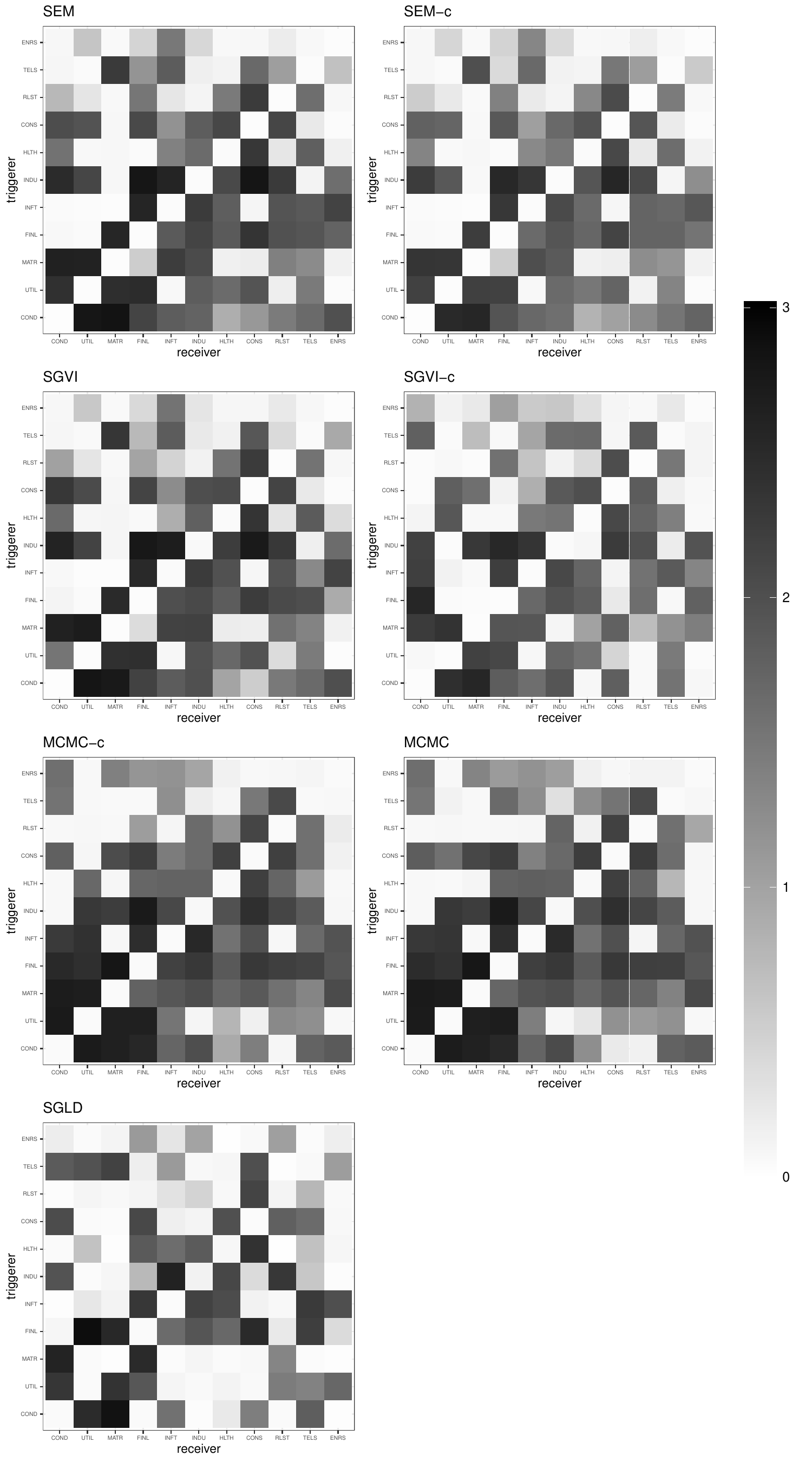}
	\caption{Heatmaps for point estimates of $\boldsymbol{\beta}$ parameters for the corresponding 11 sectors, for all seven algorithms. \label{fig:mat_beta}}
\end{figure*}

\begin{figure*}[t]
	\centering
	\includegraphics[width = \linewidth]{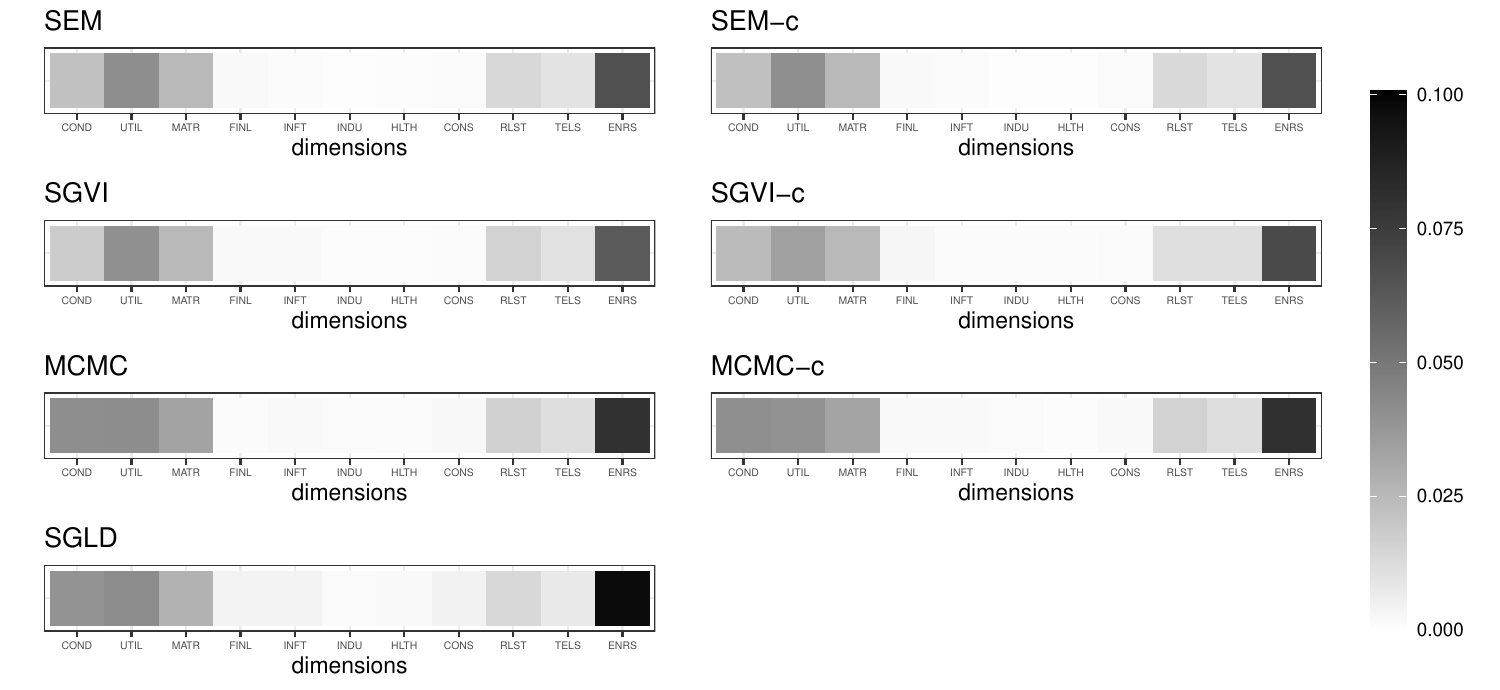}
	\caption{Heatmaps for point estimates of $\boldsymbol{\mu}$ parameters for the corresponding 11 sectors, for all seven algorithms.\label{fig:mat_mu}}
\end{figure*}

\begin{figure*}[t]
	\centering
	\includegraphics[width = \linewidth]{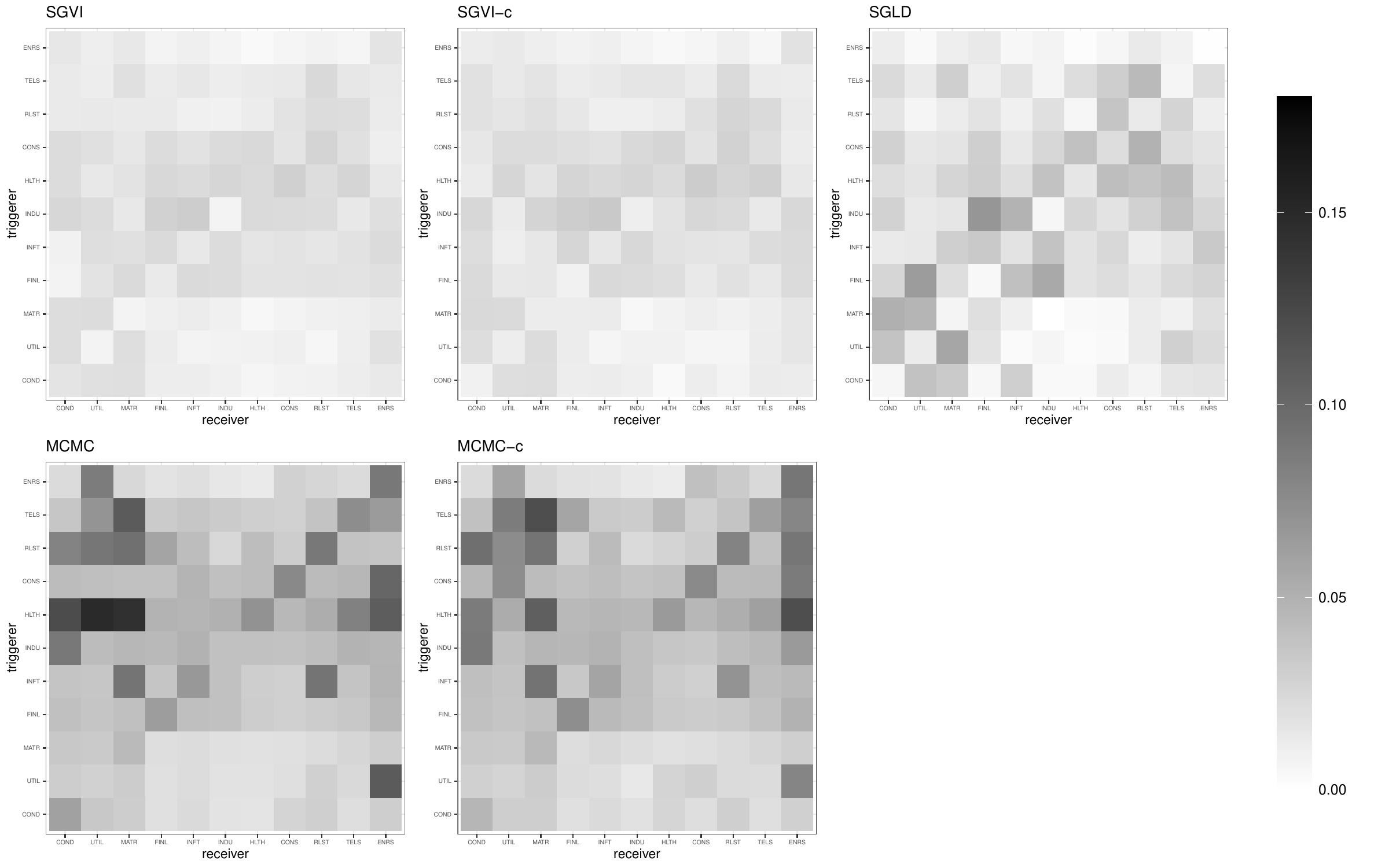}
	\caption{Heatmaps for 95\% credible interval lengths estimates of $\boldsymbol{\alpha}$ parameters for the corresponding 11 sectors, for all five algorithms that yield uncertainty estimates.
 \label{fig:app_unc_alpha}}
\end{figure*}

\begin{figure*}[t]
	\centering
	\includegraphics[width = \linewidth]{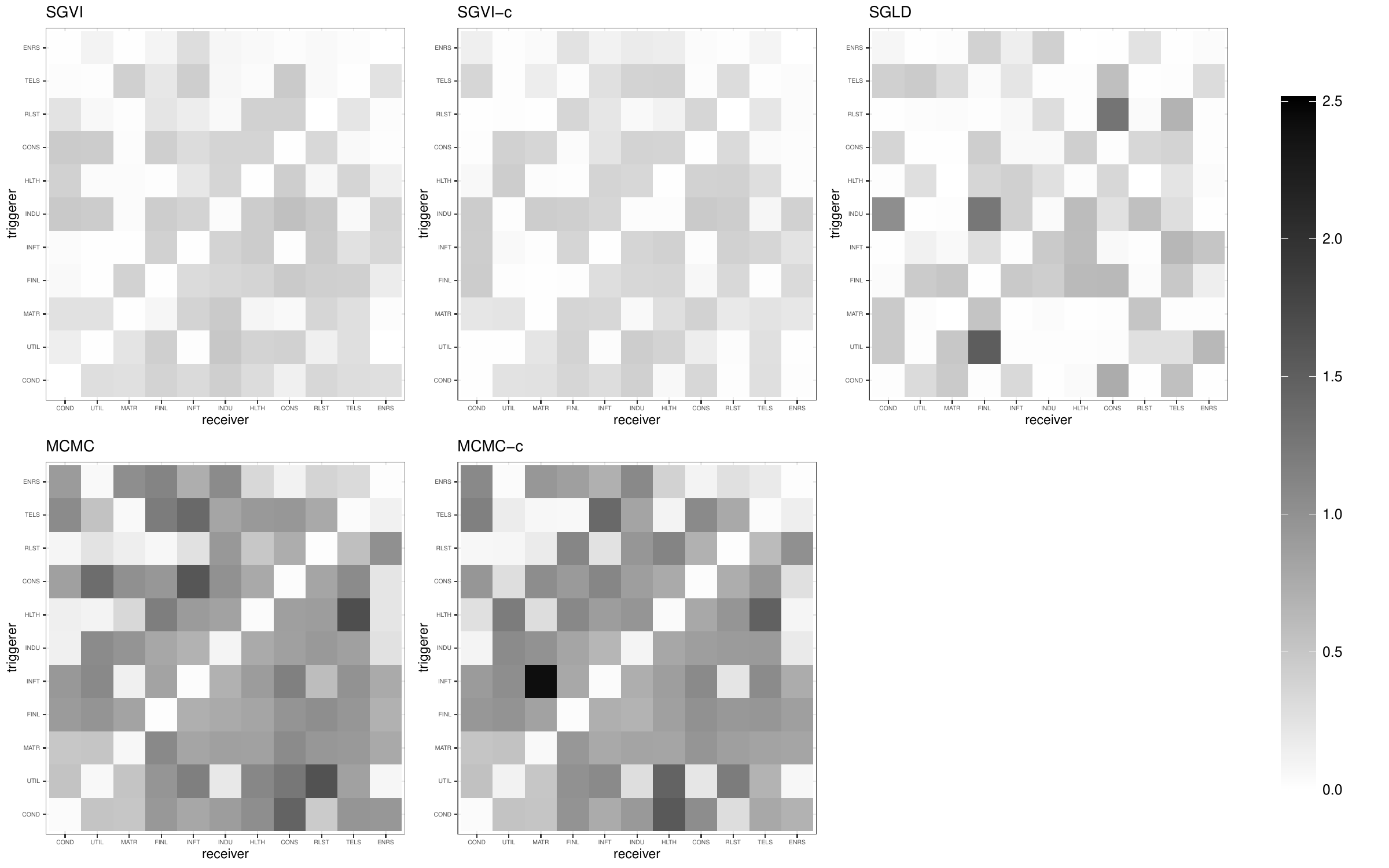}
	\caption{Heatmaps for 95\% credible interval lengths estimates of $\boldsymbol{\beta}$ parameters for the corresponding 11 sectors, for all five algorithms that yield uncertainty estimates.\label{fig:app_unc_beta}}
\end{figure*}

\begin{figure*}[t]
	\centering
	\includegraphics[width = \linewidth]{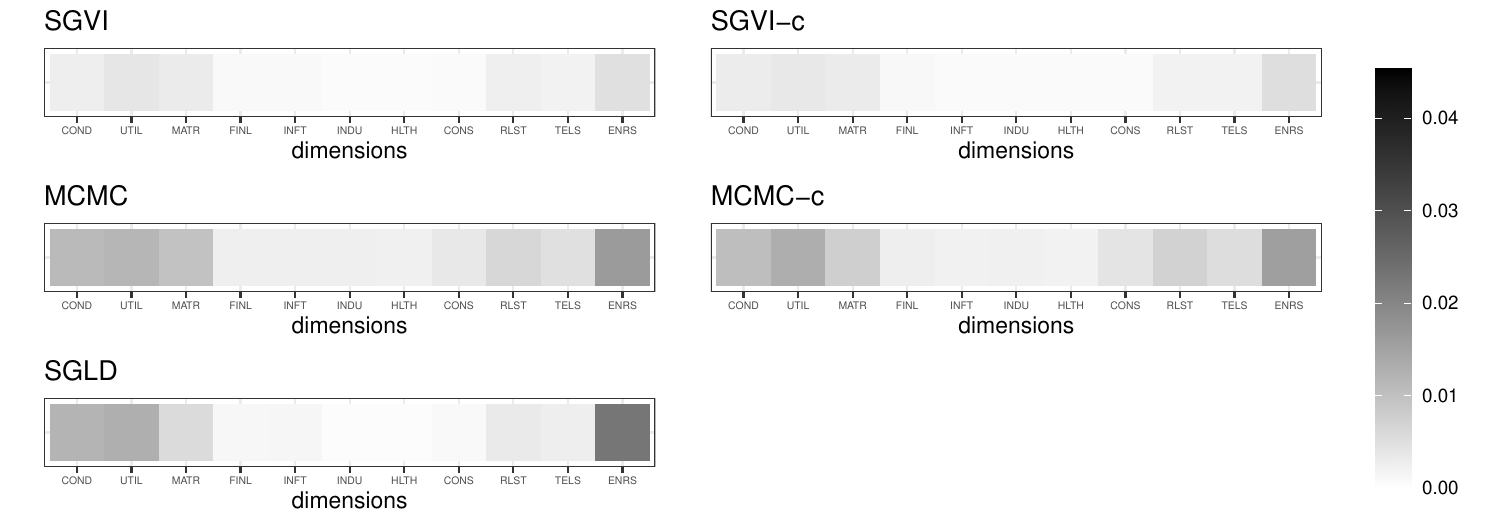}
	\caption{Heatmaps for 95\% credible interval lengths estimates of $\boldsymbol{\mu}$ parameters for the corresponding 11 sectors, for all five algorithms that yield uncertainty estimates.\label{fig:app_unc_mu}}
\end{figure*}

%{\small
%\bibliography{articles}
%}

\end{document}